\documentclass[fontsize=11pt, headings=normal, headinclude=true, paper=a4, pagesize, cleardoublepage=current, BCOR=10mm, fleqn, numbers=noenddot]{scrartcl}
\pdfoutput=1
\usepackage{lmodern}

\setkomafont{sectioning}{\normalcolor\bfseries}%Serifen
\recalctypearea % Nach ändern der Schriftart von KOMA-Script ganz bequem den Satzspiegel neu berechnen lassen.
\setcapindent{0pt}
  \usepackage[english]{babel}  % mehrsprachiger Textsatz
  \usepackage[numbers,square]{natbib}
 \usepackage{amsmath}
  \usepackage{amssymb}

  \usepackage{enumitem}
  \usepackage[utf8x]{inputenc}    % Passen Sie die Option an.
  \usepackage[T1]{fontenc}         % Schriftenkodierung
  \usepackage{graphicx}            % zum Einbinden von Grafiken
  \usepackage{color}		%Farben, f\"ur Kommentare
  \usepackage[english=british]{csquotes} %autostyle
  \usepackage{amsthm}
  
\newcommand{\R}{\mathbb{R}}
\newcommand{\C}{\mathbb{C}}
\newcommand{\N}{\mathbb{N}}
\newcommand{\Hio}{\mathfrak{H}}
\newcommand{\Fock}{\mathfrak{F}}
\newcommand{\Fou}{\mathcal{F}}
\newcommand{\coh}{\varepsilon}

\newcommand{\He}{\mathrm{H}^1}
\newcommand{\Hz}{\mathrm{H}^2}
\newcommand{\Hzo}{\mathrm{H}_0^2}
\newcommand{\Hzm}{\mathrm{H}^{-2}}
\newcommand{\Col}{\mathcal{C}}
\newcommand{\Lz}{\mathrm{L}^2}
\newcommand{\Sp}{S}
\newcommand{\Id}{\mathbf{ 1}}
\newcommand{\D}{\mathrm{d}}
\newcommand{\I}{\mathrm{i}}
\newcommand{\E}{\mathrm{e}}
\newcommand{\be}{\begin{equation}}
\newcommand{\ee}{\end{equation}}
\newcommand{\bp}{\begin{proof}}
\newcommand{\ep}{\end{proof}}
\newcommand{\epsi}{\varepsilon}
\newcommand{\ph}{\varphi}
\newcommand{\Ddhmz}{D^*_{ \Hzm}}
\newcommand{\Rminusn}{\Omega_n}
\newcommand{\Rminusnme}{\Omega_{n-1}}

\newcommand{\co}{\varepsilon(}
\newcommand{\A}{A}
\newcommand{\B}{B}
\newcommand{\cloud}{\phi}
\newcommand{\seq}{\chi}
\newcommand{\lapadj}{\Delta^*_1}
\newcommand{\lapo}{\Delta_1}
\newcommand{\lapadjn}{\Delta_n^*}
\newcommand{\lapon}{\Delta_n}

\newcommand{\lapadjf}{\Delta_{\Fock}^*}
\newcommand{\Ddelta}{D(\Delta_1^*)}

\theoremstyle{plain}
\newtheorem{thm}{Theorem}
\newtheorem{lemma}[thm]{Lemma}
\newtheorem{prop}[thm]{Proposition}
\newtheorem{corol}[thm]{Corollary}
\theoremstyle{remark}
\newtheorem{remark}{Remark}

\title{Particle Creation at a Point Source by Means of Interior-Boundary Conditions}
\author{
Jonas Lampart,\thanks{CNRS \& Laboratoire Interdisciplinaire Carnot de Bourgogne
	(UMR 6303), 
	9 Av.\ Alain Savary, 21078 Dijon, France; E-mail: 
	jonas.lampart@u-bourgogne.fr}\ \ 
Julian Schmidt,\footnote{Fachbereich Mathematik, Eberhard-Karls-Universit\"at, 
	Auf der Morgenstelle 10, 72076 T\"ubingen, Germany; E-mails: 
	julian.schmidt@uni-tuebingen.de, stefan.teufel@uni-tuebingen.de, roderich.tumulka@uni-tuebingen.de}\ \ 
Stefan Teufel,$^\dagger$\ \\ and 
Roderich Tumulka$^\dagger$ }
\date{April 25, 2017}
\begin{document}
\maketitle
\begin{abstract}
 We consider a way of defining quantum Hamiltonians involving particle creation and annihilation based on an \textit{interior--boundary condition} (IBC) on the wave function, where the wave function is the particle--position representation of a vector in Fock space, and the IBC relates (essentially) the values of the wave function at any two configurations that differ only by the creation of a particle. Here we prove, for a model of particle creation at one or more point sources using the Laplace operator as the free Hamiltonian, that a Hamiltonian can indeed be rigorously defined in this way without the need for any ultraviolet regularization, and that it is self-adjoint. We prove further that introducing an ultraviolet cut-off (thus smearing out particles over a positive radius) and applying a certain known renormalization procedure (taking the limit of removing the cut-off 
while subtracting a constant that tends to infinity) yields, up to addition of a finite constant, the Hamiltonian defined by the IBC.
 
\medskip

\noindent 
MSC: 
81T10,  % Model quantum field theories
81Q10, % Selfadjoint operator theory in quantum theory, including spectral analysis
47F05. %Partial differential operators
Key words: 
 Ultraviolet divergence problem; renormalization in quantum field theory; self-adjoint Hamiltonian; self-adjoint extensions of the Laplace operator; particle-position representation; ultraviolet cut-off. 
 
\end{abstract}

\section{Introduction}

 Interior--boundary conditions (IBCs) provide a method of defining Hamiltonian operators with particle creation and annihilation that has received little attention so far. An interesting property of this method is that, at least for some models, the common problem of ultraviolet (UV) divergence is absent. In this paper, we present rigorous results about this approach for a specific non-relativistic model of quantum field theory which show that the UV problem is indeed absent, as the Hamiltonian $H=H_\mathrm{ IBC}$ is rigorously defined and self adjoint although the sources of particle creation are point-shaped.

The UV problem,
in the form relevant to us, is the following. In the Fock space formulation of quantum field theories, the Hamiltonian involves annihilation and creation operators  $a(\seq)$ and $a^*(\seq)$ that annihilate or create   particles with wave function $\seq$. For square-integrable functions $\seq$ these operators are densely  defined  operators on Fock space. However, in most physically relevant field theories the particles are created and annihilated at points in space, and the function $\seq$ should thus be a Dirac $\delta$-distribution.   While   $a(\delta)$   can still be given mathematical sense as a densely defined operator, this is no longer possible for $a^*(\delta)$. In some cases one can take a limit of removing the ultraviolet cut-off; that is, one considers a sequence of square-integrable functions $\seq_n$ approaching the $\delta$ distribution, $\seq_n\to \delta$, and the sequence $H_{\seq_n}$ of Hamiltonians defined using $a(\seq_n)$ and $a^*(\seq_n)$ instead of $a(\delta)$ and $a^*(\delta)$ may approach 
a limit, possibly after subtraction of  suitable divergent sequence of constants $E_n$:
\begin{equation}
H_{\seq_n}-E_n \to H_{\infty}\quad \text{ as }n\to\infty\,.
\end{equation}
Then $H_\infty$ is called the renormalized Hamiltonian (see, e.g., \cite{Der03}). For a broader discussion of the  UV  problem, see, e.g.,   \cite{vH52,Lee54, Schw61,GJ85,GJ87} and also Section~\ref{sec:reno}. 

The IBC approach allows the direct definition of a Hamiltonian $H_\mathrm{ IBC}$ corresponding to $\seq=\delta$  without a renormalization procedure. It starts out from the particle--position representation of a vector in Fock space as a wave function on a configuration space of a variable number of particles. In this representation, the absorption of particle 1 by particle 2 corresponds to a jump from a configuration with 1 at the same location as 2 to the configuration without 1, while the emission of a particle corresponds to the opposite jump. These processes are therefore related to the flux of probability into (or out of)  the set $\mathcal{C}$ of collision configurations in configuration space (i.e., the configurations with two particles at the same location). As we will show, a non-trivial such flux is possible for wave functions satisfying a suitable boundary condition, with $\mathcal{C}$ regarded as the boundary of configuration space; the relevant boundary condition is a relation between the values of the wave function at the two 
configurations connected by the jump just mentioned; since it relates a boundary point to an interior point of another sector, we call this condition an \textit{interior--boundary condition}  (IBC).  
 One thus forgoes the use of creation and annihilation operators in this approach, while still obtaining non-conservation of particle number.  
Since wave functions in the domain of the Hamiltonian satisfy the IBC, the domain is not the same as that of a free field Hamiltonian. In fact, the only common element of these domains is the zero vector.
As a consequence, IBC Hamiltonians cannot be obtained as perturbations of free field Hamiltonians in any simple way.

While we discuss more general  situations in \cite{TeTu14}, we focus in our present rigorous study on the simple  model of a single non-relativistic scalar field whose quanta are created or annihilated 
 at one or more point sources at fixed locations.  
For a single source at the origin, the formal  expression for the Hamiltonian reads
\begin{equation}\label{Hdelta}
H_\delta = 
H_0+ g \bigl(a(\delta) +   a^*(\delta) \bigr)\,,
\end{equation}
where the free Hamiltonian $H_0$ is the second quantization of the non-relativistic 1-particle Hamiltonian $h=-\Delta +E_0$,  $E_0$ is a real constant called the rest energy, and $g$ is a real coupling constant.
 Note that when speaking about boundary conditions, we make  essential use of the fact that $H_0$ is sector-wise a differential operator. 
The model \eqref{Hdelta} can be regarded as a non-relativistic variant of   the Lee model \cite{Lee54},  Schweber's scalar field model \cite[Sec.~12a]{Schw61}, or  the Nelson model \cite{Nel64}.

Here we show that the IBC Hamiltonian for our model is indeed rigorously defined, self-adjoint, and (if $E_0\geq 0$) bounded from below. While it is not a perturbation of some free Hamiltonian, we show that it is equal, up to a finite additive constant, to a Hamiltonian $H_\infty$ obtained through renormalization. 
 While $H_\infty$ for the model \eqref{Hdelta} was known before to exist and can even be diagonalized explicitly, 
 an explicit characterization of its domain and its action thereon was not available.
 Thus, one conclusion from our results is that quantum field Hamiltonians obtained through renormalization can have a  simple and  explicit form when expressed in the particle--position representation, albeit not in terms of creation and annihilation operators but in terms of IBCs. And they are  no longer defined  on the domain of the free operator  $H_0$.

As a mathematical problem we have to study an infinite system of inhomogeneous boundary value problems, where the boundary on each sector is the union of  codimension-three planes. A particular difficulty arises from the fact that, in sectors of Fock space with more than one particle, these planes intersect. This makes the regularity issues more complicated, and general approaches to elliptic problems with boundaries of higher codimension (e.g.,~\cite{Ma91}) cannot be applied directly. The intersections of these planes play an important role in the theory of point interactions involving more than two particles, see~\cite{Min11, Cor12, Cor15, MiOt17, MoSe17}. See also Remark~\ref{remarkgeneral} at the end of Section~\ref{sec:symmetry} for the relation of our results to the theory of abstract boundary value problems (e.g.,~\cite{BM14}).
In our case, some of the technical difficulties associated with the boundary value problem could be circumvented if we contented ourselves with proving merely \textit{essential} self-adjointness, as we do for the generalized models of Section~\ref{sec:variants}. However, in that case we do not obtain an explicit characterization of the domain of self-adjointness. Moreover, we hope that the enhanced understanding of these boundary value problems provided by our direct approach will prove useful when dealing with further variants of the IBC approach and   point interactions.

The plan of the paper is as follows: In Section~\ref{sec:model} we motivate and define the IBC Hamiltonian $(H_\mathrm{ IBC}, D_\mathrm{ IBC})$ and state the main theorem 
 about its self-adjointness for a single point source at the origin.  
In Section~\ref{sec:reno} we discuss the relation of the IBC Hamiltonian to a Hamiltonian obtained from a standard renormalization procedure. In Section~\ref{sec:variants} we explain that our results also apply to the situation of \textit{several} (finitely many) point sources that can emit and absorb particles, located at fixed points in $\R^3$. Furthermore, we also provide  in Section~\ref{sec:variants}  a discussion of a 4-parameter family of  IBCs.
 In Sections~\ref{sec:symmetry}--\ref{sec:generalized} and the Appendix, we provide the proofs:  In Section~\ref{sec:symmetry} we prove symmetry of $H_\mathrm{ IBC}$, in Section~\ref{sec:ESA} (essential) self-adjointness, and in Section~\ref{sec:generalized} 
 we treat the generalizations of Section~\ref{sec:variants}.
  
Let us end the introduction with remarks on related literature. 
IBCs have been considered in the past, in some form or another, in \cite{LP30,Mosh51a,Mosh51b, Mosh51c,Tho84,ML91, Yaf92,Tum04}. Recent and upcoming works exploring various aspects of IBCs include \cite{TeTu14,TT15b,KS16,Gal16,bohmibc}. Introductory presentations of the kind of models considered here can be found in \cite{TeTu14,TT15b}, and the physical motivation is discussed in \cite{TeTu14}. Landau and Peierls \cite{LP30} obtained IBCs when trying to formulate quantum electrodynamics in the particle--position representation, although their Hamiltonian was still ultraviolet divergent (and thus mathematically ill defined). Moshinsky~\cite[Sec.~III]{Mosh51a} considered (as an effective description of nuclear reactions) a model with IBCs that is essentially equivalent to ours 
 (including the 4-parameter family of IBCs discussed in Section~\ref{sec:variants}), 
except that he considered only the sectors with $n=0$ and $n=1$ particles; he did not provide rigorous results about the Hamiltonian.
Yafaev~\cite{Yaf92} independently considered the same model (again only the sectors with $n=0$ and $n=1$ particles) and proved that the Hamiltonian is well defined and self-adjoint. Thomas~\cite{Tho84} considered a model analogous to ours with moving sources, but only (what corresponds to) the sectors with $n=2$ and $n=1$ particles \cite[Sec.~III]{Tho84}, respectively \cite[Sec.~II]{Tho84} with $n=1$ and $n=0$ particles, proving self-adjointness of the corresponding Hamiltonian. 
Moshinsky and Lopez~\cite{ML91} proposed a non-local kind of IBC for the Dirac and Klein--Gordon equations.
Tumulka and Georgii~\cite[Sec.~6]{Tum04} considered IBCs for boundaries of codimension 1 (whereas the boundary relevant here has codimension 3) and did not provide rigorous results. 
 Keppeler and Sieber \cite{KS16} described a physical reasoning leading to IBCs and discussed IBCs in 1 space dimension (though not rigorously). Galvan \cite{Gal16} suggested another approach towards a well defined Hamiltonian that has strong parallels to the IBC approach.
 
 The mathematical study of Hamiltonians with IBCs is closely related to that of point interactions, a field that has recently received renewed attention. Hamiltonians for $N$-particle systems with point interactions were constructed rigorously using quadratic forms by Correggi, Dell'Antonio, Finco, Michelangeli, Teta~\cite{Cor12, Cor15} and by Moser, Seiringer~\cite{MoSe17}.
 The problem was approached from the point of view of self-adjoint extensions by Minlos~\cite{Min11} and more recently by Michelangeli and Ottolini~\cite{MiOt17} (see also references therein for a more complete bibliography).

\section{The IBC Hamiltonian}

\label{sec:model}
We model the emission and absorption of non-relativistic particles at a point in $\R^3$, which we choose to be the origin. We thus call the origin the ``source'' and may think of it as a different kind of particle (which however remains at a fixed location).
 
Let $\Hio:= \Lz(\mathbb{R}^3)=\Lz(\mathbb{R}^3,\mathbb{C})$ be the one-particle Hilbert space, 
$\Hio^n:= \mathrm{ Sym} \,  \Hio^{\otimes n}$ its $n$-fold symmetric tensor product, and
$\Fock:= \Gamma(\Hio)=\bigoplus_{n \in \mathbb{N}_0} \Hio^n$ with $\Hio^0:=\C$ the symmetric Fock space over $\Hio$.
An element  $\psi$  of $\Fock$ has the form $\psi= (\psi^{(0)}, \psi^{(1)}, \psi^{(2)}, \ldots)$ with 
\be
\psi^{(n)} = \psi^{(n)}(x_1,\ldots, x_n) \in \Lz(\R^{3n})
\ee 
symmetric under permutations of its arguments and $\sum_{n=0}^\infty \|\psi^{(n)}\|^2_{\Hio^n} <\infty$.
For a bounded operator $T$ on $\Hio$, an operator $\Gamma(T)$  on $\Fock$ is defined by $(\Gamma(T)\psi)^{(n)}=T^{\otimes n}\psi^{(n)}$, and for a self-adjoint operator $h$ (possibly unbounded), 
we  define $\D\Gamma(h)$ as the generator of $\Gamma(\E^{-\I t h})$. Its action is given by
\be
( \D \Gamma (h) \psi)^{(n)} = \sum_{j=1}^n h_j \psi^{(n)}\,,
\ee 
where $h_j = \Id \otimes\ldots\otimes h\otimes \ldots \otimes \Id $ is $h$ acting on the $j$th factor.
From now on we reserve the symbol $h$ for the free one-particle Hamiltonian
\be
h=  (-\Delta+E_0, \Hz(\R^3) ) \, .
\ee
 
As a little digression, we point out how to set up a Hamiltonian with ultraviolet cut-off. We write $\overline{z}$ for the complex conjugate of $z\in\C$.
For $\seq\in \Hio$, the annihilation operator 
 \be
( a(\seq)\psi)^{(n)} (x_1,\ldots, x_n) := \sqrt{n+1} \int_{\R^3}\hspace{-1mm}\D x\, \overline{\seq(x)}\, \psi^{(n+1)} (x,x_1,\ldots,x_n)
 \ee 
 and its adjoint, the creation operator
 \be
 ( a^*(\seq)\psi)^{(n)} (x_1,\ldots, x_n) := \frac{1}{\sqrt{n}} \sum_{j=1}^n \seq(x_j) \,\psi^{(n-1)}(x_1,\ldots, \hat x_j,\ldots,x_n)
 \ee 
(where $\hat{\ }$ denotes omission) are densely defined, closed operators on $\Fock$ that are infinitesimally  $\D\Gamma(h)$-bounded when $E_0>0$. Thus, for $E_0 > 0$ and any coupling constant $g\in\R$, the total Hamiltonian  
 \begin{equation}\label{Hf}
 H_\seq := \D \Gamma (h) + g\,a(\seq) + g\,a^*(\seq) 
 \end{equation}
 is self-adjoint on the domain of $\D \Gamma (h)$ by the Kato--Rellich theorem. Operators of this type are known as \textit{van Hove Hamiltonians}  \cite{vH52,Schw61,Der03}. The limit $\seq\to \delta$ can only be taken by means of  a renormalization procedure, see Section~\ref{sec:reno} and  \cite{Der03, Nel64}.
 
We now explain how to construct explicitly an operator $H_\mathrm{ IBC}$ that  captures, as we believe, the physical meaning of ``$H_\delta$'' and agrees, as we will show, with the renormalized Hamiltonian up to  addition of a finite  constant.
Recall that with the free Schr\"odinger evolution generated by the Laplacian on $\mathrm{L}^2(\R^3)$  there is associated a probability current 
\be
j^\psi (x) = 2\, \mathrm{ Im} \,\overline{\psi(x)}\,\nabla \psi(x)\,.
\ee 
In order to allow for annihilation or creation of particles at the origin, a non-vanishing probability  current into or out of the origin must be possible. 
Using spherical coordinates $r=|x|$ and $\omega=\frac{x}{|x\vert}\in \Sp^2 = \{v\in\R^3:|v|=1\}$, this current is
\begin{eqnarray}
j^\psi_0 
&:=&  2 \lim_{r\to 0} \int_{\Sp^2}\hspace{-1mm}\D\omega\,   r^2\, \omega  \cdot \mathrm{ Im}\,{ \overline{\psi (r\omega)  }} \,\nabla\, \psi (r\omega)\nonumber \\
&=&  2  \lim_{r\to 0} \int_{\Sp^2} \hspace{-1mm}\D\omega\,  r^2\,\mathrm{ Im}\,{ \overline{\psi (r\omega)  }} \,\partial_r\, \psi (r\omega  ) \,.
\end{eqnarray}
However, for $j^\psi_0$ 
to be non-vanishing, $\psi$ or $\partial_r\psi$ must be sufficiently singular at the origin.  Since such singular functions are not in the standard domain $\Hz(\R^3)$ of the Laplacian, we need to consider the one-particle Laplace operator on a domain that includes singular functions that allow for non-vanishing currents into and out of the origin. Of course, such operators cannot be self-adjoint, since they cannot generate unitary groups.\footnote{Note that  operators with $\delta$-like potentials are defined in a similar way by enlarging the domain of the Laplacian, cf.\ \cite{DFT08}. However, in order to obtain a self-adjoint operator, an additional condition of the form $\lim_{r\to 0} \left( \partial_r   r  \psi(r\omega) - \alpha\,  r \, \psi(r\omega) \right) = 0$ with $\alpha\in\R$ is imposed, precisely  to ensure $j_0^\psi=0$.}   In order to obtain a self-adjoint Hamiltonian and a unitary evolution  on Fock space one thus needs to compensate  the loss of probability in one sector by a corresponding gain in another 
sector. 
This is achieved by connecting different sectors with boundary conditions.
Here, the configuration space is $\cup_{n=0}^\infty \R^{3n}$, and the ``boundary'' of its $n$-particle sector is the set
\be\label{Coldef}
\Col^n:=\Bigl\lbrace x\in \R^{3n} \;\Big\vert\; \prod_{j=1}^n \vert x_j \vert=0 \Bigr\rbrace
\ee
of those $n$-particle configurations with at least one particle at the origin. (This is the relevant set of collision configurations here; at these configurations, one of the moving particles collides with the source.) The ``interior--boundary condition'' connects the wave function $\psi^{(n)}$ on $\Col^n$ with the wave function $\psi^{(n-1)}$ one sector below. 

We now prepare for the precise definition of $H_\mathrm{ IBC}$. 
Define the operator $\Delta_n$ to be the Laplacian with domain $\Hzo(\R^{3n}\setminus \Col^n)\subset \Lz(\R^{3n})$, which is defined as the closure 
of $C_0^\infty(\R^{3n}\setminus \Col^n)$ in the $\Hz$-norm. We then set
\be 
 \big(\lapadjn ,D(\lapadjn)) \; \text{ is the adjoint of }
\big(\lapon,\Hzo(\R^{3n}\setminus \Col^n)\big)\,.
\ee 
Since $\Delta_n$ is densely defined, closed and symmetric, the adjoint $\lapadjn$ extends $\Delta_n$ and its domain is given by (cf.~\cite[Sect. X.1]{ReSi2})
\begin{equation}\label{domaindecomp}
D(\lapadjn) = D(\Delta_n) \oplus \ker(\lapadjn - \I) \oplus \ker(\lapadjn + \I) \, .
\end{equation}
We will always regard $D(\lapadjn)$ as a Banach space with the graph norm of $\lapadjn$. 
 Combining the $\lapadjn$ yields 
an operator $\lapadjf$ on Fock space, whose action is given by
\be 
 (\lapadjf \psi)^{(n)} := \lapadjn \psi^{(n)}\,,
\ee 
for those $\psi\in \Fock$ such that $\psi^{(n)}\in D(\lapadjn)$.
 
The role of the  annihilation operator $a(\delta)$ will be played by an operator $A$ that we define sector-wise on a dense domain to be specified later, 
$A : \Hio^{n+1}\supset D_{n+1}(A)  \to  \Hio^{n }$, by\footnote{Here and throughout the paper, we follow the convention, in order to write fewer brackets, that a derivative operator acts on all factors to the right of it, not just the one immediately to the right, unless otherwise indicated by brackets. Thus, in \eqref{Adef}, $\partial_r$ acts also on $\psi$.}
\begin{equation}\label{Adef}
( A\psi)^{(n)} (x_1,\ldots, x_n) := \frac{\sqrt{n+1}}{4\pi} \lim_{r\to0} \partial_r \,r  \int_{\Sp^2}\hspace{-1mm}\D\omega\, \psi^{(n+1 )} (r\omega, x_1,\ldots, x_n)\,.
\end{equation}
As mentioned, some $\psi^{(n+1)}(r\omega,\ldots)$ in the domain of $H_\mathrm{ IBC}$ will diverge like $1/r$ as $r\to 0$. It is not difficult to see that for $\psi^{(n+1)}$ that does not diverge as $r\to 0$, i.e.,  
for $\psi^{(n+1)} \in \Hz(\R^{3(n+1)} ) \cap \Hio^{ n+1} $,
\be 
( A\psi)^{(n)} (x_1,\ldots, x_n) = \sqrt{n+1}\, \psi^{(n+1)}(0,x_1,\ldots, x_n)\,.
\ee 
Thus, $A$ agrees with $a(\delta)$ on sufficiently regular functions. 

The boundary conditions  are formulated in terms of an operator $B$ that can again be defined sector-wise, 
$B: \Hio^{n+1} \supset D_{n+1}(B)  \to  \Hio^{n  }$, by 
\begin{equation}\label{Bdef}
( B\psi)^{(n)} (x_1,\ldots, x_n) := - 4\pi\sqrt{n+1}\: \lim_{r\to0}   \,r \, \psi^{(n+1 )} (r\omega, x_1,\ldots, x_n)\,.
\end{equation}
Again it is easy to see that for $\psi^{(n+1)} \in \Hz(\R^{3 (n+1) }) \cap \Hio^{ n+1} $ we have $( B\psi)^{(n)}=0$.

In the one-particle sector, $n=1$, the domain $\Ddelta$ is explicitly known and it is straightforward to prove that $A$ and $B$ are well defined functionals on $\Ddelta$. 
For $\gamma\in \C$ with $\mathrm{ Re}(\gamma)>0$ define the function
  \begin{align}
  \label{eq:fgamma}
  f_{\gamma}(x) 
:= - \frac{1}{4 \pi} \frac{\mathrm{e}^{-\gamma \vert x \vert}}{\vert x \vert}\,.
\end{align}
Clearly, $f_\gamma\in \mathrm{L}^2(\R^3)$ but $f_\gamma\notin \Hz(\R^3)$. Moreover, $\lapadj f_\gamma = \gamma^2 f_\gamma$ and $f_\gamma$ is the unique  $\mathrm{L}^2-$solution to this equation.
Consequently,  with \eqref{domaindecomp} it follows that  
\begin{equation}
\label{eq:vonneumann1}
D(\lapadj) = D(\lapo) \oplus V \qquad V = \mathrm{span} \Bigl\lbrace f_\gamma \Big\vert \gamma  \in \lbrace (1\pm \I)/\sqrt 2 \rbrace \Bigr\rbrace\, .
\end{equation}
Then, writing $\psi\in \Ddelta$ as $\psi_0 + \phi$ with $\psi_0\in D(\Delta_1)$ and $\phi\in V$ and integrating by parts in spherical coordinates, one finds that that the degree of asymmetry of $\lapadj$ can be expressed by $A$ and $B$, that is
\begin{equation} \label{eq:asymmn1}
 \langle \ph,  \lapadj \psi \rangle_{\Hio} - \langle  \lapadj \ph, \psi \rangle_{\Hio} = \langle B \ph , A \psi \rangle_{\C}- \langle A \ph , B \psi \rangle_{\C} \, .
\end{equation}
We will give a rigorous proof of this equation and generalize it to the case $n\geq 2$  in Propositions~\ref{prop:asymn1} and \ref{symmetry}  in Section~\ref{sec:symmetry}.
We remark that this implies that $\lapo$ has a one-parameter family of self-adjoint extensions,  known as  point interactions (cf.~\cite{AGHH88}). Their domains correspond to subspaces of $V$ on which the right hand side of Equation~\eqref{eq:asymmn1} vanishes.

To illustrate the importance of Equation~\eqref{eq:asymmn1}, we define the simplest possible IBC Hamiltonian on the truncated Fock space  $\Fock^{(1)}:=\C\oplus \Lz(\R^3)$  by 
\be
H^{(1)}_\mathrm{ IBC} := \begin{pmatrix} 0 & g A \\ 0& -\lapadj +E_0
\end{pmatrix}
\ee
on the domain 
\be
  D_\mathrm{ IBC}^{(1)} := \left\{(\psi^{(0)},\psi^{(1)}) \in \Fock^{(1)}\;\Big|\; \psi^{(1)}\in \Ddelta,\, B\psi^{(1)} = g \psi^{(0)}\right\}\,.
\ee
Here $B\psi^{(1)} = g \psi^{(0)}$ is the interior-boundary condition (IBC).
Equation~\eqref{eq:asymmn1} now implies that, contrary to 
 what it may seem like, 
$H^{(1)}_\mathrm{ IBC}$ is symmetric: for $ \ph,\psi\in D_\mathrm{ IBC}^{(1)} $
\begin{eqnarray}\nonumber
 \lefteqn{\langle \ph, H^{(1)}_\mathrm{ IBC} \psi\rangle_{\Fock^{(1)}} - \langle H^{(1)}_\mathrm{ IBC} \ph,  \psi\rangle_{\Fock^{(1)}} \;=}\\&=& -  \langle \ph^{(1)},  \lapadj  \psi^{(1)}\rangle_\Hio + \langle \lapadj  \ph^{(1)},  \psi^{(1)}\rangle_\Hio
 + \langle \ph^{(0)}, gA \psi^{1}\rangle_\C - \langle gA \ph^{(1)},\psi^{(0)}\rangle_\C\nonumber\\
 &\stackrel{\eqref{eq:asymmn1}}{=}&
 \langle A \ph^{(1)} , B \psi^{(1)} \rangle_{\C} -  \langle B \ph^{(1)} , A \psi^{(1)} \rangle_{\C} 
 + g \langle \ph^{(0)},  A \psi^{1}\rangle_\C - g\langle  A \ph^{(1)},\psi^{(0)}\rangle_\C\nonumber\\
 &\stackrel{\mathrm{ IBC}}{=}&
 g\langle  A \ph^{(1)},\psi^{(0)}\rangle_\C - g \langle \ph^{(0)},  A \psi^{1}\rangle_\C 
 + g \langle \ph^{(0)},  A \psi^{1}\rangle_\C - g\langle  A \ph^{(1)},\psi^{(0)}\rangle_\C\nonumber\\
 &=& 0\,.
\end{eqnarray}
It is not difficult to see (and was also shown in \cite{Yaf92}) that $H^{(1)}_\mathrm{ IBC} $ is even self-adjoint.

Our main result states that also the natural extension    of $H^{(1)}_\mathrm{ IBC} $ to the whole Fock space is (essentially) self-adjoint.

\begin{thm}
\label{thm:main}
For every $g, E_0\in \R$ the operator
\begin{equation}
\label{eq:H8}
 H_\mathrm{ IBC} := -\lapadjf + \D\Gamma(E_0) + gA 
\end{equation}
is essentially self-adjoint on the domain
\begin{equation}\label{eq:IBC def}
D_\mathrm{ IBC} := \left\{ \psi\in \Fock\;\left|
\begin{split}
&\psi^{(n)}\in D(\lapadjn)\cap \Hio^n \text{ for all }n\in\N\,,\\
&H \psi\in\Fock\,,\: A \psi\in\Fock\,,\:\text{and} \ B \psi  = g \psi 
\end{split}
\right.\right\}\,.
%  D_\mathrm{ IBC} := \Big\{ \psi\in \Fock\;\Big|\; \psi^{(n)}\in D(\lapadjn)\cap \Hio^n \text{ for all }n\in\N\,,\: H \psi\in\Fock\,,\: A \psi\in\Fock\,,\:\text{and} \ B \psi  = g \psi \Big\}\,.
\end{equation}
Furthermore, for $E_0 > 0$ the domain of self-adjointness equals $D_\mathrm{ IBC}$, and for $E_0 \geq 0$ the Hamiltonian $H_\mathrm{ IBC}$ is bounded from below.
\end{thm}

Note that the first two conditions in~\eqref{eq:IBC def} just ensure that $H$ maps the domain $D_\mathrm{ IBC} $ back into Fock space. The third condition, $A\psi\in\Fock$, might be redundant and follow from the second one, but we cannot show that. The last condition,
\begin{equation}
\label{eq:IBC8}
B \psi =g \psi\,,
\end{equation}
is the interior-boundary condition, which connects the limiting behavior of $\psi^{(n)}$ at the boundary of the $n$-particle sector (where one particle reaches the origin) with the wave function $\psi^{(n-1)}$ one sector below. 

Formally, an analogous computation to the one for $H^{(1)}_\mathrm{ IBC}$ shows that $H_\mathrm{ IBC}$ is symmetric (see the proof of Corollary~\ref{corol:focksymmetry}). 
However, in order to establish Equation~\eqref{eq:asymmn1} for $n\geq 2$, we need to first investigate the regularity of functions in the adjoint domain $D(\lapadjn)$. This will be carried out in Section~\ref{sec:symmetry}, with the main result given by Proposition~\ref{symmetry}. The proof of (essential) self-adjointness in Section~\ref{sec:ESA} uses the symmetry established in Section~\ref{sec:symmetry} and a comparison with a renormalized operator to be defined below.

\section{The connection to renormalization}\label{sec:reno}

As mentioned already, the formal expression $H_\delta$ as in \eqref{Hdelta} can be regularized by means of an ultraviolet cut-off, then the cut-off can be removed (while constants $E_n$ tending to $\pm\infty$ get subtracted) in order to obtain a renormalized Hamiltonian $H_\infty$. Our main result in this section, Theorem~\ref{thm:renorm}, asserts that $H_\mathrm{ IBC}$ agrees with $H_\infty$ (up to addition of a finite constant relative to the standard choice of $E_n$). We state Theorem~\ref{thm:renorm} in Section~\ref{sec:renodef} and then put it into perspective in Section~\ref{sec:renotechnique} by connecting it to known facts, techniques, and hitherto open questions about $H_\infty$.

\subsection{Definition of $H_\infty$ and relation to $H_\mathrm{ IBC}$}
\label{sec:renodef}

We approximate the formal Hamiltonian $H_\delta$ with regularized (cut-off) Hamiltonians 
\begin{equation}
H_n= \D\Gamma(h) + g \Bigl( a(\seq_n) +   a^*(\seq_n) \Bigr) = H_0 + H_{I_n} 
\end{equation}
with any choice of $\seq_n\in \Lz(\R^3)$ such that $\seq_n\to\delta$ as $n\to\infty$ in the sense that $\hat \seq_n \rightarrow  \hat \seq_\infty:=  \hat{\delta} = (2\pi)^{-3/2}$ pointwise with $\|\hat \seq_n\|_\infty$ uniformly bounded. Here  $\mathcal{F}\seq=\hat \seq$ denotes the Fourier transform of $\seq \in \Lz(\R^d)$. It is easy to see using standard arguments (and will be explained below) that if $E_0>0$ then $H_n-E_n$ converges in the strong resolvent sense for
\begin{equation}\label{Endef}
E_n := - g^2 \langle  \seq_n , {h}^{-1}  \seq_n \rangle_{\Lz}\,.
\end{equation}
Note that for $E_0>0$ the free one-particle operator $h= -\Delta +E_0 \geq E_0>0$ is invertible.  
The limit is called the renormalized Hamiltonian,
\begin{equation}
\label{eq:Hhatdef}
H_\infty:=\lim_{n\to\infty} (H_n-E_n)\,.
\end{equation}

\begin{thm}\label{thm:renorm}
For $E_0>0$, the renormalized operator $(H_\infty, D(H_\infty))$ agrees with  $(H_\mathrm{ IBC},  D_\mathrm{ IBC})$ up to  an additive constant: 
\be
D_\mathrm{ IBC} = D(H_\infty)
\quad \text{and} \quad
  H_\mathrm{ IBC}   = H_\infty +  \frac{g^2 \sqrt{E_0}}{4\pi}\,\mathbf{ 1}_\Fock\,.
\ee 
The spectrum of $H_\mathrm{ IBC}$ is given by $\lbrace E_{min} \rbrace \cup [E_{min}+E_0, \infty)$ and $E_{min}=g^2\sqrt{E_0}/4\pi$ is a simple eigenvalue.\\
Moreover, for $g\neq 0$ we have that $D(H_\infty) \cap D(H_0^{1/2}) = D_\mathrm{ IBC} \cap D(\D\Gamma({h}^{1/2}))  = \lbrace 0 \rbrace$.
\end{thm}

Theorem~\ref{thm:renorm} is established in Section~\ref{sec:ESA}.

\subsection{Remarks on the renormalization procedure}
\label{sec:renotechnique}

The above described renormalization scheme is a particularly simple case of a somewhat more general
 renormalization procedure 
that can be applied  to a wider class of UV divergent Hamiltonians with the following common structure. There is a self-adjoint operator   $(H_0,D(H_0))$ and a sequence of operators $H_{I_n}$ that are small perturbations of $H_0$ in the sense that
\be 
H_n := H_0 + H_{I_n} 
\ee 
is self-adjoint on $D(H_0)$. If the interaction operator $H_{I_n}$ converged as $n\to\infty$ to an operator that is relatively (form-)bounded by $H_0$ with relative bound smaller than one, then no renormalization would be necessary. In a typical manifestation of the UV problem, however, $H_{I_n}$ does not converge. But in the cases of interest, there is a sequence of numbers $E_n\to\pm\infty$ such that
$H_\infty=\lim_{n\to\infty} (H_n - E_n)$
exists in the strong resolvent sense.

In the examples we have in mind, the essential steps in finding this sequence $E_n$ and proving the convergence of $H_n-E_n$ are, first, to construct a certain sequence of unitary operators $W_n$ on Fock space, called dressing transformations, such that $W_nH_n W_n^*$ assumes a manageable form; second, to split $W_nH_n W_n^*$ into  
\begin{equation}
\label{eq:Hfhat}
W_n H_n  W^*_n = H_n' +E_n
\end{equation}
such that $H_n'$ converges   in the strong resolvent sense  to a well defined operator $H'_\infty$. 
Third, one shows that $W_n$ has a strong limit $W_\infty$ (which is automatically unitary). Then it follows that
\begin{equation}
H_n-E_n = W_n^* H_n' W_n \xrightarrow{n\to\infty} W_\infty^* H_\infty' W_\infty = H_\infty
\end{equation}
in the strong resolvent sense.

Depending on the concrete model, the determination  of the limiting Hamiltonian $H'_\infty = \lim_{n\to\infty}H'_n$ can be more or less tricky and, as a consequence, its domain can be more or less explicit. In all examples discussed in the following, $W_n$ leaves invariant the domain $D(H_0)$, but this is no longer true for $W_\infty$. 

In his seminal paper \cite{Nel64}, Nelson showed that the model nowadays named after him can be renormalized according to the general scheme just sketched. He used the so-called Gross transformation for $W_n$ and was able to 
characterize $(H_\infty', D(H_\infty'))$ 
as a form perturbation of $H_0$. Hence, he could not explicitly determine $D(H'_\infty)$ but merely conclude that $D(H_\infty') \subset D(H_0^{1/2})$. 

Whenever $H_\infty'$ is an operator-bounded perturbation of $H_0$, one has $ D(H_\infty') = D(H_0)$ and $D(H_\infty)=W_\infty^*D(H_0)$ can be determined through the mapping properties of $W_\infty^*$.
Recently, Griesemer and W\"unsch \cite{GrWu16a} proved that the Fr\"ohlich Hamiltonian, which describes polarons, is of that type.
 In this case, one can 
define $H_\infty$ also directly via its quadratic form without the detour via the dressing transformation. However, then the domain of $H_\infty$ remains unknown,
while the result of \cite{GrWu16a} provides an explicit characterization of it.   
In our model \eqref{Hdelta},
the situation is even simpler, since it turns out that $H_n'=H_\infty'=H_0$.

After the existence of a self-adjoint renormalized Hamiltonian $H_\infty$ is established, two questions remain in general open. First, is there a direct characterization of  the domain $D(H_\infty) = W^*_\infty D(H_\infty')$? And second, how does $H_\infty$ act explicitly? As Nelson \cite{Nel64} put it:
\begin{quote}
It would be interesting to have a direct description of the operator $H_\infty$. Is $D(H_\infty) \cap D(H_0^{1/2}) = 0$?
\end{quote} 
The answer to the last question 
has been given by Griesemer and W\"unsch for the Fr\"ohlich Hamiltonian in \cite{GrWu16a} and for the massive  Nelson model in \cite{GrWu16b}. 
For our model \eqref{Hdelta}, we answer both of Nelson's questions in Theorem~\ref{thm:renorm}
in terms of $(H_\mathrm{ IBC}, D(H_\mathrm{ IBC}))$.

Here is what the dressing transformation $W_n$ looks like for our model \eqref{Hdelta}. 
Since ${h}^{-1}  \seq_n \in \Lz(\R^3)$ for $n \leq \infty$, the field operator
\begin{equation}
\label{eq:fieldop}
\Phi({h}^{-1}  \seq_n) := a({h}^{-1} \seq_n ) + a^*({h}^{-1}  \seq_n )
\end{equation}
is self-adjoint. Therefore,
\begin{equation} 
W_n := \E^{-\I \Phi(\I g {h}^{-1} \seq_n )} 
\end{equation}
is unitary for all $n\leq \infty$. It is straightforward to show that~\eqref{eq:Hfhat} now holds with
$E_n$ as in \eqref{Endef} and $H_n' :=  \D \Gamma (  h) $.  The proof can be found in Section~\ref{subsec:renorm},  or, for example, also in \cite{Deck04,Der03}.
Then $\lim_{n\to \infty} E_n= -\infty$, and $H'_\infty=\lim_{n\to \infty} H_n'=  \D \Gamma ( h)$ clearly exists. As a consequence,
\be
H_\infty = W^*_\infty \, \D \Gamma (   h) \, W_\infty \qquad\mbox{on}\quad D(H_\infty) = W^*_\infty D(\D \Gamma (   h))\,.
\ee

\section{Variants of the IBC Hamiltonian}
\label{sec:variants}

\subsection{General interior-boundary conditions}

The IBC $B\psi = g\psi$ discussed in the previous sections is not the only possibility of implementing  interior-boundary conditions for the Laplacian.  
In this section we present a four-parameter family of different interior-boundary conditions  that all lead to a self-adjoint Hamiltonian on Fock space. In a certain sense, 
this family covers all possible 
 types of IBCs.

The wider class of IBCs involves, instead of the values of the wave function on the boundary (like a Dirichlet boundary condition), a linear combination of the values and the derivative of the wave function on the boundary (like a Robin boundary condition); such IBCs were formulated in \cite{TeTu14,TT15b} for boundaries of codimension 1 (and are also considered in \cite{STT16a} for particle creation, where the boundary has codimension 3). Specifically, in this wider class, we replace
\be
B \to \E^{\I\theta}(\alpha B + \beta A)\,,\quad
A \to \E^{\I\theta}(\gamma B + \delta A)\,,
\ee
where $\theta\in[0,2 \pi)$ and $\alpha,\beta,\gamma,\delta\in\R$ are such that
\be
\alpha\delta - \beta \gamma =1\,,
\ee
 so that four of the five parameters can be chosen independently. We absorb the coupling constant $g$ into the constants $\alpha,\beta,\gamma,\delta$.  That is, we replace the IBC $B\psi=g\psi$ by 
\begin{equation}\label{IBC9}
\E^{\I\theta}(\alpha B + \beta A)\psi = \psi
\end{equation}
and the Hamiltonian
$
{H}_\mathrm{ IBC} = -\lapadjf +\D\Gamma(E_0)+gA
$ by
\be\label{H9}
\tilde{H}_\mathrm{ IBC} =-\lapadjf +\D\Gamma(E_0)+ \E^{\I\theta}(\gamma B + \delta A)\,.
\ee
The previous IBC \eqref{eq:IBC8} and Hamiltonian \eqref{eq:H8} are obviously contained in this scheme by chosing   $\theta=0=\beta=\gamma$ and $\alpha^{-1}=g=\delta$. As discussed in detail in \cite{STT16a}, the phase~$\theta$ can be removed by means of the gauge transformation $\psi^{(n)}\to \E^{-\I\theta n}\psi^{(n)}$ if there is a single source, but not if there are several sources with different $\theta$'s, a situation that we consider in the next section. We refrain from stating and proving the analogue to Theorem~\ref{thm:main}  also for  $\tilde H_\mathrm{ IBC}$, although it could be proved along the same lines as for $H_\mathrm{ IBC}$.  Instead, Theorem~\ref{thm:variant1} below implies already a statement that is merely slightly weaker, namely that, for $E_0>0$, $\tilde H_\mathrm{ IBC}$ is essentially self-adjoint on a dense domain satisfying the IBC  \eqref{IBC9}.

 To which extent does  the family $\tilde H_\mathrm{ IBC}$ cover all possible Hamiltonians  
 with IBCs?  
Yafaev \cite{Yaf92} showed that for the model on the truncated Fock space $\C\oplus \Lz(\R^3)$ with either zero or one particle all possible extensions of the (not densely defined) operator 
\be
H^\circ = ( 0, - \Delta) \quad\mbox{on}\quad D(H^\circ) = \{0\}\oplus C_0^\infty(\R^3\setminus \{0\})
\ee
are of the above type. On Fock space, however, one has in principle much more freedom. We could connect different sectors by different IBCs, i.e., make $\theta, \alpha,\beta,\gamma,\delta$ all depend on $n$, or even let them depend on the configuration of the other particles.  But if we exclude such a dependence, then Yafaev's result shows that the family $\tilde H_\mathrm{ IBC}$ is complete.
 
\subsection{IBCs for multiple sources}

We now consider a finite number $N$ of sources fixed at (pairwise distinct) locations $\xi_1,\ldots,\xi_N\in \R^3$. 
To keep things simple, we assume $E_0>0$ for the remainder of this section. For each source $\xi_i$, $1\leq i \leq N$, we choose parameters
\be
v_i := (\theta_i, \alpha_i, \beta_i , \gamma_i, \delta_i) \in [0, 2\pi) \times \R^4
\ee
which fullfill separately
\begin{equation}
\label{vcond1}
\alpha_i \delta_i - \beta_i \gamma_i = 1 \qquad 1\leq i \leq N \,.
\end{equation} 
We write $v$ for $(v_1,\ldots,v_N)$.
For suitable $\psi\in\Hio$, define
\be
A_i \psi := \lim_{x \rightarrow \xi_i} \partial_{r_i} (r_i \psi(x))\,, \qquad
B_i \psi := - 4 \pi \lim_{x \rightarrow \xi_i} (r_i \psi(x))\,, \qquad
\text{where }  r_i :=\vert x - \xi_i \vert \,,
\ee
and
\be
X_i :=\E^{\I \theta_i}(\alpha_i B_i + \beta_i A_i)\,, \qquad
Y_i :=\E^{\I \theta_i}(\gamma_i B_i  + \delta_i A_i)\,, \qquad 1\leq i\leq N \,.
\ee
The corresponding Fock space operators
\be
X_i^\Fock \big \vert_{\Hio^{n+1}} := \sqrt{n+1}\, X_i \otimes \Id_{\Hio^n}\,, \qquad
Y_i^\Fock \big \vert_{\Hio^{n+1}} := \sqrt{n+1}\, Y_i \otimes \Id_{\Hio^n}
\ee
are densely defined in $\Fock$. Then $(\lapadj , \Ddelta) := (\lapo , C_0^\infty(\R^3 \setminus \lbrace \xi_1,\xi_2 , \dots , \xi_N \rbrace))^*$ is a closed but non-symmetric operator on $\Hio$. Nevertheless, we will use the symbol $\D \Gamma(-\lapadj)$ to denote the operator which acts as $-\sum_{j=1}^n \Id_{1, \dots , j-1} \otimes \lapadj \otimes \Id_{j+1, \dots ,n}$ on the $n$-th sector of Fock space. It is well known  \cite{AGHH88,DFT08}  that 
\be \tilde h := -\lapadj + E_0 \quad \text{on} \quad U(v):= \bigcap_{i=1}^N \mathrm{ ker} X_i \subset D(\lapadj) \ee
is a self-adjoint operator that is bounded from below. It is called the $N$-center point interaction with energy offset $E_0$ and parameters $a_i := \frac{\alpha_i}{\beta_i}$, where $\beta_i=0$ corresponds to $a_i = + \infty$.
\begin{thm}
\label{thm:variant1}
Let $E_0 >0$ and $v$ be any set of parameters obeying the condition \eqref{vcond1} given above. There exists a dense subspace $\tilde{D}_\mathrm{ IBC} \subset \Fock$ such that for $\psi \in \tilde{D}_\mathrm{ IBC}$ the IBCs
\be X_i^\Fock \psi = \psi \qquad \forall   \ 1 \leq i \leq N   \ee 
hold and such that
\be
\tilde{H}_\mathrm{ IBC} := \D\Gamma(-\lapadj+E_0) + \sum_{i=1}^N Y_i ^\Fock 
\ee
is essentially self-adjoint on $\tilde{D}_\mathrm{ IBC}$.
 If $\tilde h$ is strictly positive\footnote{i.e., there is a positive constant $c$ such that $\tilde h\geq c$.}, then $\tilde H_\mathrm{ IBC}$ is bounded from below and possesses a unique ground state. 
\end{thm}
\begin{remark} Suppose that $\beta_i = 0$ for all $1\leq i \leq N$. Then 
\be \tilde h = h = (-\lapadj + E_0, \Hz(\R^3)) \ee
is the free one-particle operator, which is strictly positive. In this case $\tilde H_\mathrm{ IBC}$ is bounded from below for any choice of distinct points $\xi_1, \dots, \xi_N$.
\end{remark}
\begin{remark}
Let $N=1$. In this case, for all values of $a_1 = \frac{\alpha_1}{\beta_1} \in (- \infty, \infty]$, the essential spectrum of the point-interaction operator is $\sigma_{ess}(\tilde h) = [ E_0, \infty)$, cf. \cite{AGHH88}. If $a_1 \geq 0$, then $\tilde{h}$ has no point spectrum. If $a_1 < 0$, then there is exactly one eigenvalue $\lambda_0$ of $\tilde h$. It is explicitly given as $ \lambda_{0} = E_0 - 16 \pi^2 a_1 ^2$. Therefore $\tilde H_\mathrm{ IBC}$ is bounded from below if $a_1  > 0$ or if $a_1  \leq 0$ but still $a_1  > - \frac{\sqrt{E_0}}{4 \pi}$. 
\end{remark}
Under certain assumptions on $v$ and $E_0$, we are able to further characterize $\tilde H_\mathrm{ IBC}$. In order to state the theorem, we have to introduce some abbreviations: \\ For any $\lambda >0$ let 
\be
w^{\lambda}_i (x) :=  f_{\sqrt{\lambda}}(x-\xi_i) = - \frac{\E^{-\sqrt{\lambda} \vert x -\xi_i \vert}}{{4 \pi} \vert x - \xi_i \vert}  \in \Lz(\R^3)\,,
\ee
and define the   matrices 
\be
G^\lambda_{ij}:=  w^\lambda_i(\xi_j) =  w_j^\lambda(\xi_i)\,,
\ee
and 
\be
S_{ij}(\lambda) := \delta_{ij} \E^{\I \theta_i}\Bigl(\alpha_i + \frac{\sqrt{\lambda}}{4 \pi} \beta_i\Bigr) +(1- \delta_{ij})  \E^{\I \theta_i} \beta_i G^\lambda_{ij} \,,
\ee
where $\delta_{ij}$ denotes the Kronecker symbol. 
 Note that $S$ depends on all of $\lambda,\xi_1,\ldots,\xi_N$, $v_1,\ldots,v_N$.
\begin{thm}
\label{thm:variant2}
Let $(\tilde{H}_\mathrm{ IBC} , \tilde{D}_\mathrm{ IBC})$ also denote the unique self-adjoint extension that has been constructed in Theorem \ref{thm:variant1}.
If   the vector $(1,1, \dots, 1)^{T}$ lies in the range of $S(E_0)$, then there exists $\cloud \in D(\lapadj) \subset \Hio$ such that we have the equality 
\begin{align}
\label{unitaryeq}
 \E^{ \I \Phi(\I \cloud)} \, \tilde H_\mathrm{ IBC} \, \E^{- \I \Phi(\I \cloud)} =  \D \Gamma(\tilde h) + C(\cloud) \Id_\Fock 
\end{align}
as self-adjoint operators on Fock space $\Fock$. Here $C(\cloud) \in \R$ is a constant, $\Phi$ has been defined in \eqref{eq:fieldop} and $\D\Gamma(\tilde h)$ denotes the second quantization of $\tilde h = (-\lapadj +E_0,U)$.
\end{thm}
The definition of $\tilde{D}_\mathrm{ IBC}$ in terms of coherent states obtained from vectors in $\Ddelta$, as well as the proof of Theorems~\ref{thm:variant1} and~\ref{thm:variant2} and the explicit form of the ground state, of $\cloud$ and of $C(\cloud)$ are given in Section~\ref{sec:generalized}. As discussed in detail in \cite{STT16a}, $\tilde H_\mathrm{ IBC}$ is time reversal invariant if and only if all $\theta_i$ coincide up to addition of an integer multiple of $\pi$.

\section{Symmetry of $H_\mathrm{ IBC}$ }
\label{sec:symmetry}

In this section we prove symmetry of $(H_\mathrm{ IBC}, D_\mathrm{ IBC})$. The main ingredient is Equation~\eqref{eq:asymmn1}, which will be proved in Proposition~\ref{prop:asymn1} below, and its generalization to $n\geq 2$.

\begin{prop}
\label{prop:asymn1}
For $n=1$ the maps $A$ and $B$ are well-defined continuous linear functionals on $\Ddelta$ and for any $\ph, \psi \in \Ddelta$ we have
\be 
 \langle \ph,  \lapadj \psi \rangle_{\Hio} - \langle  \lapadj \ph, \psi \rangle_{\Hio} = \langle B \ph , A \psi \rangle_{\C}- \langle A \ph , B \psi \rangle_{\C} \, .
\ee 
\end{prop}
\begin{proof}
Recall that $\Ddelta= D(\lapo) \oplus V$ with $V= \mathrm{ span} \left\{ f_\gamma\,\Big|\, \gamma\in \{ (1\pm\I)/\sqrt{2} \}\right\}$. On the functions $f_\gamma$ one easily evaluates 
\be
A f_\gamma = \frac{\gamma}{4\pi} \quad\mbox{and}\quad B f_\gamma =1 \,.
\ee
On $D(\lapo)$ we have $A=0$, since for $\psi\in C^1(\R^3)$
\be
A\psi = \frac{1}{4\pi} \lim_{r\to 0}   \int_{\Sp^2}\left( \psi(r\omega )  + r\omega\cdot\nabla\psi(r\omega)\right) \, \D \omega = \psi(0)\,,
\ee
and the point evaluation is continuous on $D(\lapo)= \Hzo(\R^3\setminus\{0\})$. Clearly also $B=0$ on $D(\lapo)$. Now since $\Hzo$ is a closed subspace of $\Ddelta$, the projection $p:  \Ddelta \to \Ddelta/ \Hzo \cong V$ is continuous.  Thus $A,B: \Ddelta\to \C$ are continuous as they can be written as the composition of $p$ with a linear functional on a finite dimensional space.

The difference on the left hand side of~\eqref{eq:asymmn1} vanishes if either $\varphi$ or $\psi$ are elements of $\Hz_0(\R^3\setminus \lbrace 0\rbrace)$, and so does the right hand side by the considerations above. Thus, it is sufficient to verify the claim for $\ph=f_{\gamma_1}$, $\psi=f_{\gamma_2}$. As noted before we have $\lapadj f_\gamma=\gamma^2 f_\gamma$ and
\be 
 \langle f_{\gamma_1}, f_{\gamma_2} \rangle = \frac{1}{4\pi}\int_0^\infty \D r\,\E^{-(\overline\gamma_1+\gamma_2)r}
 =\frac{1}{4\pi(\overline\gamma_1+\gamma_2)}\,.
\ee 
Thus 
\begin{eqnarray}
 \langle f_{\gamma_1}, \lapadj f_{\gamma_2} \rangle - \langle \lapadj f_{\gamma_1}, f_{\gamma_2} \rangle
&=&\frac{\gamma_2^2- \overline\gamma_1^2}{4\pi(\overline\gamma_1+\gamma_2)}  
\; =\;\frac{\gamma_2 -\overline\gamma_1}{4\pi} \nonumber\\[1mm]&=& \overline{B f_{\gamma_1}}Af_{\gamma_2}- \overline{A f_{\gamma_1}}Bf_{\gamma_2} \, .  
\end{eqnarray}
\end{proof}

Proposition~\ref{prop:asymn1} can be understood as a generalized integration-by-parts formula for the singular functions in $D(\lapadj)$. Its generalization to the case $n \geq 2$, given in Proposition~\ref{symmetry} below, requires knowledge of the regularity properties of functions in $D(\lapadjn)$. These are rather subtle, as the following example shows:

Let $f \in \mathrm{H}^{-1/2}(\R^3)$, and set 
\begin{equation}\label{eq:sing_ex}
 \psi(x,y)=- \frac{\E^{T |x|}}{4 \pi |x|}f(y)\,, 
\end{equation}
where $\E^{T |x|}$ denotes the contraction semi-group with generator $T=-\sqrt{-\Delta_y+1}$,  $D(T)= \He(\R^3)$, acting on $\Lz(\R^3_y)$. One easily checks that $\psi \in \Lz(\R^6)$ with norm proportional to $\|f\|_{\mathrm{H}^{-1/2}}$. By the smoothing properties of the semi-group, $\psi$ is a smooth function on $\R^6\setminus\{x=0\}\supset \R^6\setminus\Col^2$. The action of $\Delta^{*}_2$ on $\psi$ is thus given by differentiating on $\R^6\setminus\Col^2$ and yields
\be 
 \Delta^{*}_2 \psi= \psi\,,
\ee 
so $\psi\in D(\Delta^*_2)$ is an eigenfunction of $\Delta^{*}_2$ with eigenvalue one. However, applying only the differential expression $\Delta_x$ gives $\Delta_x \psi=T^2 \psi$, which is not an element of $\psi \in \Lz(\R^6)$ unless $f\in \mathrm{H}^{3/2}(\R^3)$. Thus we have $\psi \in D(\Delta^*_2)$, but applying the Laplacian in only one of the variables does not give a square-integrable function, i.e.~$\psi\notin D(\Delta_1^* \otimes 1)$.
Furthermore, the formula for $\psi$ suggests that $B\psi=\sqrt{2} f \in \mathrm{H}^{-1/2}(\R^3)$ is a distribution, so the ``boundary values'' of $\psi$ on the collision configurations $\Col^2$ will be of low regularity.

We now state our results concerning the definition of the operators $A$ and $B$ on $D(\lapadjn)$, which we prove in Appendix~\ref{sec:appreg}. 
To allow for a lighter notation, we will use the symbol $\Omega_n$ to denote the configuration space of $n$ particles, that is $\Omega_n := \R^{3n} \setminus \Col^n$.

\begin{lemma}
\label{lem:ABcont}
For any $n\in\N$, every $\ph\in D( \lapadjn)$ has a representative for which the limits
\be
A^{(n)} \ph  := \frac{\sqrt{n}}{4\pi} \lim_{r\to 0}  \partial_r \int_{\Sp^2}r \ph(r\omega ,x_2,\ldots, x_n) \, \D \omega 
\ee
and 
\be
B^{(n)} \ph := -4\pi \sqrt{n} \lim_{r\to 0} r \ph(r\omega ,x_2,\ldots, x_n)
\ee
 exist in  $\Hzm(\Rminusnme)$ and this defines continuous linear maps 
\be 
 A^{(n)} ,B^{(n)} :  D( \lapadjn)\to \Hzm(\Rminusnme)\,.
\ee 
Furthermore, $B^{(n)}$ vanishes on $\mathrm{H}^1(\R^{3n}) \cap D(\lapadjn)$ and the restriction of $A^{(n)}$ to $\Hz(\R^{3n})$ is given by the Sobolev-trace on $\{x_1=0\}$.
\end{lemma}

In the following we will drop the superscript from $A^{(n)}$ and $B^{(n)}$ for better readability. 
Let
\begin{equation}\label{Dn*def} 
D^*_n := \Bigl\{\psi\in D(\lapadjn)\cap \Hio^n \,\Big|\, A \psi \in \Lz(\R^{3n-3})\,,\; B \psi\in \Lz(\R^{3n-3})\Bigr\} \subset \Hio^n \, .
\end{equation}
and equip this space with the norm $\Vert \psi \Vert_{\Hio^n}  + \Vert \lapadjn \psi \Vert_{\Hio^n}+ \Vert A \psi \Vert_{\Hio^{n-1}} + \Vert B \psi \Vert_{\Hio^{n-1}}$. The following Proposition characterizes $\Hz\subset D_n^*$ in terms of boundary values.

\begin{prop}\label{regprop}
Let $\ph\in D^*_n$. Then $B \ph=0$ if and only if  $\ph\in \Hz(\R^{3n})$.
\end{prop}

With this a-priori information on the functions in $D_n^*$ we can now characterize the asymmetry of $\lapadjn$ in terms on the operators $A$ and $B$.

\begin{prop}
\label{symmetry}
For all $\psi, \varphi \in D^*_n$ we have that
\begin{equation}\label{eq:symmetry}
\langle \lapadjn \psi, \varphi \rangle_{\Hio^{n}}  - \langle  \psi, \lapadjn \varphi \rangle_{\Hio^{n}} = \langle A \psi, B \varphi \rangle_{\Hio^{n-1}} -\langle B \psi, A \varphi \rangle_{\Hio^{n-1}}  \,.
\end{equation}
\end{prop}

\begin{proof}
By definition of the norm on $D^*_n$, the maps $A,B: D^*_n\to \Hio^{n-1} $ are continuous, and so is the map 
\be 
 \mathfrak{B} : D^*_n \to \Hio^{n-1}\oplus \Hio^{n-1}\,,\qquad \psi\mapsto (B\psi, A\psi)\,. 
\ee 
The skew-hermitean sesquilinear form
\be 
 \beta(\psi, \varphi):=\langle \lapadjn \psi, \varphi \rangle - \langle  \psi, \lapadjn \varphi \rangle
\ee 
is also continuous on $D^*_n$.  Suppose for the moment that there exists a continuous, skew-hermitean sesquilinear form $\alpha$ on $\mathrm{ran } \, \mathfrak{B} \subset \Hio^{n-1} \oplus \Hio^{n-1}$ such that $\beta = \alpha\circ \mathfrak{B}$. 
 Any continuous sesquilinear form  on $\mathrm{ran }\, \mathfrak{B}$ is already determined by its values on any subspace of $\mathrm{ran } \, \mathfrak{B}$ which is dense in the $\Vert \cdot \Vert_{n-1} + \Vert \cdot \Vert_{n-1}$-norm. Therefore, $\beta$ is already determined by its values on a subspace $D_0$ whose image $\mathfrak{B}(D_0)$ is dense in $\Hio^{n-1} \oplus \Hio^{n-1}$.  That is, it suffices to verify \eqref{eq:symmetry} on $D_0$.  Such a subspace is given by
\be D_0 := \lbrace \psi \in D^*_n \vert \psi = \psi_A + \psi_B, \,  \psi_{A} \in D^n_A , \, \psi_{B} \in D^n_B  \rbrace  \qquad D_{A/B} := \ker A/B \subset \Ddelta \, .
\ee
Here $D^n_A$ and $D^n_B$ are the spans of symmetric $n$-fold tensor products of elements of $\ker A$ and $\ker B$ on $\Ddelta$. These kernels are the domains of self-adjoint extensions of $\lapo$; in fact $\ker B = \Hz(\R^3)$, and $\ker A$ is the domain of a point source with infinite scattering length. We have
\be \mathfrak{B}(D_0) = (B(D^n_A),A(D^n_B)) = \left( D^{n-1}_{A}, D^{n-1}_{B}\right) \subset \Hio^{n-1} \oplus \Hio^{n-1} \, ,\ee
so $\mathfrak{B}(D_0)$ is in fact dense. Now because $(\lapadj , D_{A})$ and $(\lapadj , D_{B})$ are symmetric operators and $\beta$ is skew-hermitean, it is even sufficient to compute only one cross-term $\beta(\psi_A, \ph_B)$. For tensor products, however, Proposition \ref{prop:asymn1} can be applied and yields
\begin{align}
\beta(\psi_A, \varphi_B)&=\sum_{i=1}^n  \langle (\lapadj)_{x_i} \psi_A,\varphi_B \rangle_{\Hio^n} - \langle  \psi_A, (\lapadj)_{x_i}\varphi_B \rangle_{\Hio^n}\nonumber\\
  &= n\left( \langle (\lapadj)_{x_1} \psi_A, \varphi_B \rangle_{\Hio^n} - \langle  \psi_A, (\lapadj)_{x_1}\varphi_B\rangle_{\Hio^n}\right) \nonumber\\
  &=\langle A \psi_A , B \varphi_B\rangle_{\Hio^{n-1}}-\langle B \psi_A , A \varphi_B \rangle_{\Hio^{n-1}} \nonumber  \\
  &= - \langle B\psi_A, A\varphi_B\rangle_{\Hio^{n-1}} \, .
\end{align}
We still have to construct an $\alpha$ with $\beta = \alpha\circ \mathfrak{B}$. Here Proposition~\ref{regprop} enters as the key ingredient: we have that
\begin{equation}
 \ker \mathfrak{B}= \ker B\cap\ker A =\lbrace \psi \in \Hz(\R^{3n})\cap \Hio^n\vert A\psi=\psi\vert_{\Col^{n}}=0\rbrace
=\Hzo(\Rminusn) \,. \label{eq:kerneloffrakb}
\end{equation} 
As a consequence $\beta(\psi, \varphi)=0$ for all $\varphi \in D^*_n$ if $\psi\in \ker \mathfrak{B}$.
Thus we can define on the quotient the sesquilinear form
\be
\tilde \alpha: D^*_n/ \ker \mathfrak{B}\times  D^*_n/ \ker \mathfrak{B}\to \C\,,\quad ([\psi],[\ph])\mapsto \beta(\psi,\ph)\,
\ee
and \eqref{eq:kerneloffrakb} guarantees that this is well defined. Let $\pi$ denote the quotient map. Then $\beta = \tilde{\alpha} \circ \pi
$, which means that $\tilde{\alpha}$ is continuous in the quotient topology. There exists a unique continuous isomorphism  $\mathfrak{B}': D^*_n/ \ker \mathfrak{B} \rightarrow \mathrm{ran}\, \mathfrak{B}$ such that $\mathfrak{B} = \mathfrak{B}' \circ \pi$. Inserting the identity we get
\be \beta = \tilde{\alpha} \circ \pi = \tilde{\alpha} \circ (\mathfrak{B}')^{-1} \circ \mathfrak{B}' \circ \pi = \tilde{\alpha} \circ (\mathfrak{B}')^{-1} \circ \mathfrak{B}  \,.  \ee
If we define $ \alpha := \tilde{\alpha} \circ (\mathfrak{B}')^{-1}$, it is obviously continuous. This proves the claim.
\end{proof}
\begin{corol}
\label{corol:focksymmetry}
$(H_\mathrm{ IBC}, D_\mathrm{ IBC})$ is symmetric for all $E_0 \in \R$.
\end{corol}
\begin{proof}
Recall the definition of the domain
\be 
D_\mathrm{ IBC} := \left\{ \psi\in \Fock\;\left|
\begin{split}
&\psi^{(n)}\in D(\lapadjn)\cap \Hio^n \text{ for all }n\in\N\,,\\
&H \psi\in\Fock\,,\: A \psi\in\Fock\,,\:\text{and} \ B \psi  = g \psi 
\end{split}
\right.\right\}\,.
\ee
Now $H \psi \in \Fock$ together with $A \psi \in \Fock$ clearly implies $(-\lapadjf+ \D \Gamma(E_0)) \psi \in \Fock$, so we may split the operator and compute with the help of Proposition \ref{symmetry}:
\begin{align}
\langle \ph , H \psi \rangle_\Fock \nonumber
&= \langle \ph , (-\lapadjf + \D \Gamma(E_0)) \psi \rangle_\Fock + \langle \ph ,  g A \psi \rangle_\Fock \\
&= \sum_{n\in \N} \langle \ph^{(n)} , -\lapadjn  \psi^{(n)} \rangle_{{n}} + \langle \ph , \D \Gamma(E_0) \psi \rangle_\Fock+  \langle \ph , g A \psi \rangle_\Fock &\nonumber\\
&\hspace{-3pt}\stackrel{\eqref{eq:symmetry}}{=}
\begin{aligned}[t]
 &\sum_{n\in \N} \langle - \lapadjn  \ph^{(n)} ,  \psi^{(n)} \rangle_{{n}} +   \langle A \ph^{(n)}, B \psi^{(n)} \rangle_{{n-1}} - \langle B \ph^{(n)} , A \psi^{(n)} \rangle_{{n-1}}\\
 & +\langle \ph , \D \Gamma(E_0) \psi \rangle_\Fock+  \langle \ph , g A \psi \rangle_\Fock
\end{aligned}
\nonumber\\
&\hspace{-3pt}\stackrel{\text{IBC}}{=} 
\begin{aligned}[t]
&\langle (\lapadjf + \D \Gamma(E_0)) \ph,  \psi \rangle_\Fock + \langle \ph , g A \psi \rangle_\Fock \\
&+ \sum_{n\in \N}  \langle A \ph^{(n)}, g\psi^{(n-1)} \rangle_{{n-1}} - \langle g \ph^{(n-1)} , A \psi^{(n)} \rangle_{{n-1}}
\end{aligned}
\nonumber \\
&=  \langle (\lapadjf + \D \Gamma(E_0)) \ph,  \psi \rangle_\Fock + \langle g A \ph, \psi\rangle_{\Fock} = \langle H \ph,  \psi \rangle_\Fock \, . 
\end{align}
\end{proof}
Another simple corollary of our results in this section is the fact that, for $E_0>0$ and $g\neq 0$, the intersection of $D_\mathrm{ IBC}$ and the form-domain of the free operator $\D\Gamma(h)$ contains only the zero vector. More precisely:
\begin{corol}\label{corol:form-domain}
 Let $g\neq 0$ and set $h_+=-\Delta + |E_0|$, then for any $E_0\in \R$ we have 
 \be 
  D_\mathrm{ IBC}\cap D\left(\D \Gamma (h^{1/2}_+)\right) = \{0\}\,.
 \ee 
\end{corol}
\begin{proof}
Take $\psi\neq 0 \in D_\mathrm{ IBC}$. Then $\psi^{(n)}\neq 0$ for some $n\in \N$. This implies that $B\psi^{(n+1)}=g \psi^{(n)}\neq 0$. But $D(\D \Gamma (h^{1/2}_+))\vert_{\Hio^{n+1}}=\He(\R^{3(n+1)})\cap \Hio^{n+1}$, and by Lemma~\ref{lem:ABcont} $B$ vanishes on this set, so $\psi\notin D(\D \Gamma (h^{1/2}_+))$.
\end{proof}

\begin{remark}\label{remarkgeneral}
Propositions~\ref{regprop} and~\ref{symmetry} prove that $(\Hio^{n-1}, B, A)$ is a quasi boundary triple (in the sense of~\cite{BM14}) for the operator $(-\Delta_n^*, D_n^*)$. This allows for a complete characterization of the adjoint domain $D(\lapadjn)$ and the self-adjoint extensions of $\Delta_n$ (restricted to symmetric functions $\Hio^n$).
The following statements are consequences of the general theory \cite[Prop.~2.9,~2.10]{BM14}, but can also be concluded directly in our setting from Propositions~\ref{regprop} and~\ref{symmetry}.

For any $\lambda>0$ we have that
\be 
 D(\lapadjn)\cap \Hio^n= \Hz(\R^{3n})\cap \Hio^n \oplus K_\lambda\,,
\ee 
with $K_\lambda=\mathrm{ker}(-\lapadjn + \lambda)\cap \Hio^n$. 
The map
\be 
 B:K_\lambda \to \big(\mathrm{H}^{1/2}(\R^{3(n-1)})\cap \Hio^{n-1}\big)'\subset \mathrm{H}^{-1/2}(\R^{3(n-1)})
\ee 
is continuous, as can  easily be seen from the proof of Lemma~\ref{lem:ABcont}.
 By Proposition~\ref{regprop} it is one-to-one. It is also surjective, with inverse given, as in~\eqref{eq:sing_ex}, by
\begin{align}
f\mapsto &\; \mathrm{Sym}_n\left( \frac{\E^{-\sqrt{-\Delta+1} \vert x_n \vert}}{4\pi \vert x_n \vert} f(x_1, \dots, x_{n-1}) \right)\nonumber\\
&=\mathrm{Sym}_n \left( (-\Delta + 1)^{-1} f(x_1, \dots, x_{n-1}) \delta_0 (x_n)\right)
\,.
\end{align}
Such formulas for functions in $D(\Delta_n^*)$ have been widely used in the literature on point interactions, see e.g.~\cite{Min11}. 
An alternative rigorous proof that for $n=2$ the whole adjoint domain can be obtained in this way has been published only very recently,~\cite[Prop.\ 4]{MiOt17}.
\end{remark}

\section{Essential Self-Adjointness of $H_\mathrm{ IBC}$}\label{sec:ESA}
 \subsection{Coherent Vectors and Denseness}
The aim of this subsection is to introduce a set of coherent vectors in the domain $D_\mathrm{ IBC}$   on which we can perform many computations explicitly. A standard choice of a  dense set in Fock space is the space $\Fock_0$  containing the  vectors with a bounded number of particles, i.e., $\psi\in\Fock_0$ iff there exists $N\in\N$ such that $\psi^{(n)}=0$ for $n>N$. However, $\Fock_0\cap D_\mathrm{ IBC} = \{0\}$ since the IBC $B\psi=g\psi$ immediately yields that if $\psi^{(n)}\neq 0$, then $\psi^{(k)}\neq 0$ for all $k>n$.

For $u \in \Hio$ the associated \textit{coherent vector} $\coh(u)\in\Fock$ is defined by  
\be 
\coh(u)^{(n)} := \frac{u^{\otimes n}}{\sqrt{n!}}\,.  
\ee
It holds that $\langle \coh(v),\coh(u) \rangle_\Fock = \exp({  \langle v,u\rangle_\Hio})$; thus, the nonlinear map $\coh: \Hio\to\Fock$, $u\mapsto \coh(u)$, is continuous,
\begin{eqnarray}
\label{eq:dfsghgfh}
\| \coh(v) - \coh(u)\|^2 
&=& \langle \coh(v) , \coh(v) \rangle_\Fock + \langle \coh(u) , \coh(u) \rangle_\Fock - 2 \mathrm{Re}  \left( \langle \coh(v) , \coh(u) \rangle_\Fock \right) \nonumber 
\nonumber\\[1mm]
&
=& \E^{\| v  \|^2_\Hio} +\E^{\| u  \|_\Hio^2} - 2 \mathrm{Re}\  \E^{\langle v , u \rangle_\Hio}  \xrightarrow{v\to u} 0  \,.
\end{eqnarray}
For a subset $D \subseteq \Hio$, consider the subspace spanned by coherent vectors of elements of $D$, that is
\be
E(D) := \mathrm{span}\lbrace \coh(u) \vert u \in D \rbrace\subset \Fock\,.
\ee 
We will refer to this subspace as the \textit{coherent domain over $D$.}
When working with coherent vectors, we will need the following generalized polarization identity.
\begin{prop}\label{polar}
Let $V$ be a complex vector space and   $v_1,\ldots,v_n \in V$. Then there exist vectors $u_1,\ldots, u_m \in V$ and coefficients $d_1,\ldots, d_m \in \mathbb{C}$ such that
\be 
\label{eq:polarization}
\mathrm{ Sym}\, (v_1 \otimes \cdots \otimes v_n) 
= \sum_{k=1}^m d_k \, u_k^{\otimes n}\,.
\ee
\end{prop}
See Appendix A.2 for the proof, including an explicit formula for $u_k$ and $d_k$
For a densely defined, non-self-adjoint operator $(T,D)$, we use the expression $\D \Gamma(T)$ to denote the operator which acts as $\sum_{j=1}^n \Id_{1, \dots , j-1} \otimes T \otimes \Id_{j+1, \dots ,n}$ on the $n$-th sector of Fock space. This expression obviously has meaning on $E(D)$.
\begin{prop}
\label{prop:coherentdense}
If $D \subset  \Hio$ is dense, then $E(D)$ is a dense subspace of $\Fock$.
Moreover, let $(T,D)$ be a densely defined operator on $\Hio$. Then for  $f \in \Hio$ we have
\begin{eqnarray}\label{f1}
a(f)\, \coh(u) &=& \langle f, u \rangle_{\Hio}\, \coh(u) \qquad \mbox{for all }u \in \Hio \,,
\\\label{f2}
a^*(f)\, \coh(u) &=& \frac{\D}{\D t}\bigg\vert_{t=0} \coh(u+tf) \qquad \mbox{for all }u \in \Hio  \,,
\\ \label{f3}
\D\Gamma(T)\, \coh(u) &=&   a^*(Tu) \,\coh(u) = \frac{\D}{\D t}\bigg\vert_{t=0} \coh(u+tTu) \qquad \mbox{for all }u \in D \,.
\end{eqnarray}
\end{prop}
\begin{proof}
For $u \in \Hio$  the map $\R\to \Fock$, $t \mapsto \coh(tu)$, has derivatives of any order at $t=0$ with
\be 
\left( \frac{\D^n}{\D t^n} \bigg\vert_{t=0} \coh(tu) \right)^{(m)} 
= \begin{cases} 0 \quad &m \neq n \\ \sqrt{n!} \ u^{\otimes n} \quad &m=n\,. \end{cases}
\ee
Thus, $E(\Hio)$ is dense in the span of all vectors of the form $(0, \dots , u^{\otimes n},  0 \dots)$.
Then, by the generalized polarization identity (Proposition~\ref{polar}) and standard approximation arguments, $E(\Hio)$ is also dense in $\Fock$. The continuity of the map $u\mapsto \coh(u)$ finally implies that $E(D)$ is dense in $E(\Hio)$ whenever $D$ is dense in $\Hio$. The formulas \eqref{f1}--\eqref{f3} follow directly from the definitions of the corresponding operators. 
\end{proof}

The natural candidate for the set $D$ is of course $\Ddelta$. However, we still need to make sure that the coherent vectors generated by $D$ satisfy the boundary condition. 
Let
\be 
 D_g^{\gamma}:=\Bigl\lbrace \ph\in \Hio \Big\vert \ph
=g f_\gamma + \phi, \ \phi \in \Hz(\R^3) \Bigr\rbrace
\ee 
for some $\gamma$ with $\mathrm{Re} \, \gamma >0$.
The affine subspace $D_g^\gamma$ is  dense in $\Hio$ because $\Hz(\R^3)$ is dense. Then, according to Proposition~\ref{prop:coherentdense}, the coherent domain $E(D_g^\gamma)$ over $D_g^\gamma$   is a dense subspace of $\Fock$; in fact, it is included in $D_\mathrm{ IBC}$:

\begin{corol}
\label{corol:mandd}
We have that $E(D_g^\gamma) \subset  D_\mathrm{ IBC}$ for the value of $g$ used in $D_\mathrm{ IBC}$ and any $\gamma\in\C$ with $\mathrm{Re}\, \gamma>0$. As a consequence,  $ D_\mathrm{ IBC}$ is dense in $\Fock$.
\end{corol}

\begin{proof} 
Let $\ph \in D_g^\gamma \subset \Ddelta$. Then obviously $\coh(\ph)^{(n)} \in D^*_n$ as in \eqref{Dn*def},  and
\be 
 (B \coh(\ph))^{(n)}= \sqrt{n+1} (B\ph)\frac{\ph^{\otimes n }}{\sqrt{(n+1)!}}
=g \frac{\ph^{\otimes n }}{\sqrt{(n)!}}=g\coh(\ph)^{(n)}\,,
\ee 
so $\coh(\ph)$ satisfies the interior--boundary condition. Additionally,
\begin{equation}
\label{eq:aeigenvalue}
 (A \coh(\ph))^{(n)} 
=\sqrt{n+1}(A \ph)\frac{\ph^{\otimes n }}{\sqrt{(n+1)!}}
=(A \ph)\coh(\ph)^{(n)}\,,
 \end{equation}
which defines an element of $\Fock$ since $A$ is bounded on $\Ddelta$ by Proposition \ref{prop:asymn1}. Observe that $(\lapadj)_{x_j}\coh(\ph)^{(n)} \in \Lz(\R^3_{x_j} , \Lz(\R^{3n-3}))$. Therefore the action of $\lapadjn$ coincides on $E(D_g^\gamma)$ with that of $\sum_{j=1}^n (-\lapadj)_{x_j}$. It is also straightforward to check that $\lapadjf \coh(\ph)\in\Fock$, and this completes the proof.
\end{proof}
 
 \subsection{Unitary Equivalence} 
\label{sec:renormalization}
To avoid unnecessary technicalities, we define the dressing transformation $\mathrm{e}^{-\mathrm{i}\Phi}$ directly for coherent states and not in terms of its generator $\Phi=a+a^*$. That is, we write $W(\ph)$ for $\mathrm{e}^{-\mathrm{i}\Phi(\mathrm{i} \ph)}$ and construct $W(\ph)$ as follows. For $\ph,u\in\Hio$, let 
\begin{align}
\label{eq:dressingdef}
W(\ph) \, \coh(u) := \E^{-\langle \ph,u\rangle_{\Hio} - \frac{\Vert \ph \Vert_{\Hio}^2}{2}} \, \coh(u+\ph)\,.
\end{align}

\begin{lemma}
\label{lemma:unitary}
For every $\ph \in \Hio$, the map $W(\ph)$ can be extended uniquely to a unitary transformation on Fock space; its inverse is given by $W(-\ph)$.
\end{lemma}

See, e.g., Section IV.1.9 in \cite{Mey93} for the rather elementary proof.

\begin{prop}
\label{prop:esasecquant}
Let $(T,D)$ be a self-adjoint operator on $\Hio$. Then its 
 second quantization $\D\Gamma(T)$  is essentially self-adjoint on the coherent domain $E(D)$.
\end{prop}

\begin{proof}The coherent domain $E(D)$ is a subspace of $D(\D\Gamma(T))$ and the associated unitary group of $\D\Gamma(T)$ is given by $\Gamma(\E^{-\I T t})$. Since its action on coherent vectors is extremely simple, $\Gamma(\E^{-\I T t}) \coh(u) = \coh(\E^{-\I T t} u)$, the coherent domain over $D$ is  \textit{invariant} under $\Gamma(\E^{-\I T t})$ because $D$ is. Now the statement follows from Nelson's invariant domain theorem \cite[Thm.~VIII.11]{ReSi1}.
\end{proof}

\begin{lemma}
\label{lemma:pulltrough}
Let $(T,D)$ be a densely defined operator on $\Hio$. Suppose that $\ph , u \in D$, and let $W(\ph)$ be the corresponding unitary dressing transformation defined by \eqref{eq:dressingdef}. Then
\be 
W(-\ph) \D\Gamma(T) W(\ph) \big \vert_{E(D)}
= \D\Gamma(T)+ a^*(T\ph) + a(T\ph) + G(T,\ph) \ \big \vert_{E(D)}  \, ,
\ee
where $G(T,\ph)$ is an operator on $E(D)$ whose action is given by 
\be G(T,\ph) \coh(u) = \left(\langle \ph, Tu \rangle_{\Hio}-\langle T\ph, u \rangle_{\Hio} +\langle \ph, T\ph \rangle_{\Hio} \right) \coh(u) \, . \ee
\end{lemma}
\begin{proof}
This is a consequence of  Proposition~\ref{prop:coherentdense} and the following straightforward computation:
\begin{align}
&W(-\ph)\D\Gamma(T)W(\ph)\coh(u) \\
 &\stackrel{\eqref{f3}}{=} W(-\ph)\tfrac{\D}{\D t}\big\vert_{t=0} \coh(u+\ph +t T(u+\ph)) \E^{-\langle \ph, u \rangle - \tfrac{\Vert \ph \Vert}{2}}\nonumber\\
 &\stackrel{\eqref{eq:dressingdef}}{=}\tfrac{\D}{\D t}\big\vert_{t=0}\coh(u +t T(u+\ph)) \E^{t \langle \ph, T(u+\ph)  \rangle}\nonumber\\
 &\stackrel{\eqref{f2}}{=}  \left(a^*(T(u+\ph)) + \langle \ph, Tu \rangle_{\Hio} + \langle \ph, T\ph \rangle_{\Hio} \right) \coh(u)\nonumber\\
 &\stackrel{\eqref{f3}}{=} \left(\D\Gamma(T)+ a^*(T\ph) + \langle \ph, Tu \rangle_{\Hio} + \langle \ph, T\ph \rangle_{\Hio} \right) \coh(u) \nonumber\\
 &\stackrel{\eqref{f1}}{=} \left(\D\Gamma(T)+ a^*(T\ph) + a(T\ph) + \langle \ph, Tu \rangle_{\Hio} - \langle T \ph , u \rangle_{\Hio} + \langle \ph, T\ph \rangle_{\Hio} \right) \coh(u)\,. \qedhere\nonumber
\end{align}
\end{proof}
\begin{corol}
\label{corol:ifinvertible}
Let $(T,D)$ be a self-adjoint operator on $\Hio$ which is invertible, i.e. $0 \in \rho(T)$. Then for $\psi \in \Hio$ and $u \in D$ it holds that
\be W(- T^{-1} \psi) \D \Gamma(T) W(T^{-1} \psi) \big \vert_{E(D)} = \D \Gamma(T) + a^*(\psi) + a(\psi) + \langle \psi, T^{-1} \psi \rangle_\Hio \Id_\Fock \ \big \vert_{E(D)} \ee
\end{corol}
\begin{proof}
Apply Lemma \ref{lemma:pulltrough} with $\ph = T^{-1} \psi$ and observe that, because $T$ is symmetric, it holds that $\langle \ph, T u \rangle_{\Hio} - \langle T \ph,  u \rangle_{\Hio} = 0$. So the operator $G(T, \ph)$ reduces to multiplication with the constant $\langle T^{-1} \psi,  \psi \rangle_\Hio = \langle \psi, T^{-1} \psi \rangle_\Hio$.
\end{proof}
\begin{corol}
\label{corol:allgemein}
Let $E_0 \in \R$, $\gamma>0$, $f_\gamma$ be given by \eqref{eq:fgamma} and let $h=-\Delta +E_0$ with domain $\Hz(\R^3)$. Then on the coherent domain $E(\Hz(\R^3))$ we have
\begin{align}\nonumber
&W(- g f_\gamma) \,  H_\mathrm{ IBC}  \, W(g f_\gamma) \big \vert_{E(\Hz(\R^3))} \\
&= \D\Gamma(h)+ (-\gamma^2 + E_0) \left(a^*(g f_\gamma) + a(g f_\gamma)\right) + C(g,\gamma,E_0) \Id_{\Fock}  \ \big \vert_{E(\Hz(\R^3))}
\end{align}
where the constant reads
\be C(g,\gamma,E_0) = (-\gamma^2 + E_0) \Vert g f_\gamma  \Vert^2_{\Hio}  + g^2 \frac{\gamma}{4 \pi} \, . \ee
\end{corol}
\begin{proof} We start by noting that~\eqref{eq:aeigenvalue} gives for $u\in \Hz(\R^3)$
\begin{equation}\label{AW}
 gA W(gf_\gamma)\coh(u)
 =
 g(A(gf_\gamma+u))  W(gf_\gamma) \coh(u)
 =\left(\frac{g^2 \gamma}{4\pi} + gu(0)\right)W(gf_\gamma)\coh(u)\,.
\end{equation}
Now set $(T,D)=(-\lapadj + E_0,\Ddelta)$ and $\ph = g f_\gamma$ in Lemma \ref{lemma:pulltrough}. 
Then
\begin{align}
&W(- g f_\gamma) \, H_\mathrm{ IBC}  \, W(g f_\gamma) \coh(u)  \\\nonumber
&= W(- g f_\gamma) \D\Gamma(-\lapadj + E_0) W(g f_\gamma) \coh(u) +   \left(\frac{g^2 \gamma}{4\pi} + gu(0)\right) \coh(u)\\\nonumber
&= \left(\D\Gamma(h) +  ( E_0-\gamma^2 ) \left(a^*(g f_\gamma) + a(g f_\gamma)\right)  + G(T,\ph)  +   \left(\frac{g^2 \gamma}{4\pi} + gu(0)\right) \right)\coh(u)\,.
\end{align}
It remains to show that  for $u \in \Hz(\R^3)$
\be
\left(G(T,\ph)  +  \frac{g^2 \gamma}{4\pi} + gu(0)  \right)\coh(u) = C(g,\gamma,E_0)\coh(u)\,.
\ee
It follows from Proposition~\ref{prop:asymn1} that
\begin{eqnarray}
G(T,\ph) + gu(0) &=& 
g \langle  f_\gamma , T u \rangle - g\langle T (  f_\gamma) , u \rangle + g^2 \langle f_\gamma, T f_\gamma\rangle + gAu \nonumber\\
&=&  g \langle   f_\gamma , -\lapadj  u \rangle - g \langle -\lapadj  f_\gamma , u \rangle +  gA u +    (-\gamma^2 + E_0) \Vert g f_\gamma  \Vert^2_{\Hio} \nonumber\\ 
 &=&  g A f_\gamma B u   - g B f_\gamma A u + g  A u +    (-\gamma^2 + E_0) \Vert g f_\gamma  \Vert^2_{\Hio}\nonumber  \\&=&(-\gamma^2 + E_0) \Vert g f_\gamma  \Vert^2_{\Hio}\,,
\end{eqnarray}
 since $B u =0$ and $Bf_\gamma=1$.
\end{proof}
\begin{prop} 
\label{prop:moregeneral}
For all $E_0\in \R$  the operator  $ (H_\mathrm{ IBC},D_\mathrm{ IBC})$ is essentially self-adjoint and for any $\gamma >0$ the space  $W(g f_\gamma) E(\Hz(\R^3))\subset D_\mathrm{ IBC}$ is a core. If $E_0 \geq 0$, then the Hamiltonian $H_\mathrm{ IBC}$ is bounded from below.
\end{prop}
\begin{proof}
According to Corollary \ref{corol:allgemein} and by symmetry of $ (H_\mathrm{ IBC},D_\mathrm{ IBC})$ it suffices to show that 
\begin{equation}\label{transham}
 \D\Gamma(h)+ (-\gamma^2 + E_0) \left(a^*(g f_\gamma) + a(g f_\gamma)\right) 
 \end{equation}
is essentially self-adjoint on $E(\Hz(\R^3))$. By Proposition \ref{prop:esasecquant}, the operator \\ $(\D\Gamma(h), E(\Hz(\R^3)))$ is essentially self-adjoint. 

For $E_0\geq0$ the perturbation $a^*(g f_\gamma) + a(g f_\gamma)$ is infinitesimally bounded with respect to $\D\Gamma(h)$ (see Proposition 3.8 in \cite{Der03}) and thus, by Kato-Rellich, essential self-adjointness of \eqref{transham} on $E(\Hz(\R^3))$ holds. Here one uses the fact that
\begin{align}
\label{eq:fourieroff}
\hat{f_\gamma}(k)  = - (2 \pi)^{-\frac{3}{2} } (\vert k \vert^2 + \gamma^2)^{-1} = - \hat{\delta}(k) \cdot (\vert k \vert^2 + \gamma^2)^{-1} \qquad \mathrm{Re}(\gamma) > 0\,,
\end{align}
and therefore $ \langle \hat{f}_\gamma, \hat{h}^{-1} \hat{f}_\gamma  \rangle < \infty$  even for $E_0=0$.

If $E_0 < 0$, essential self-adjointness of \eqref{transham}  is shown using  Nelson's Commutator Theorem (Theorem X.36 in \cite{ReSi2}) with comparison operator $N = \Id_\Fock + \D \Gamma(h - E_0 + 1)$, cf.\ Proposition 3.11 in
\cite{Der03}.
\end{proof}

\begin{prop}
\label{prop:haupt}
If $E_0>0$, then  the operator  $ (H_\mathrm{ IBC},D_\mathrm{ IBC})$ is self-adjoint and  
\begin{equation}\label{euqality1}
H_\mathrm{ IBC} =  W(gf_{\sqrt{E_0}})\left[\D\Gamma( h) + \tfrac{g^2 \sqrt{E_0}}{4 \pi} \right] W(-gf_{\sqrt{E_0}})  \, . \end{equation}
\end{prop}
\begin{proof}
As $E_0>0$, we may choose $\gamma = \sqrt{E_0}$ in Corollary \ref{corol:allgemein} and set $\cloud := g f_{\gamma= \sqrt{E_0}}$. The constant $C(g,\sqrt{E_0},E_0)$ then reduces to $\frac{g^2 \sqrt{E_0}}{4 \pi}$ and the equality \eqref{euqality1} holds on the common core $W(\cloud)E(\Hz(\R^3))$. This extends to the common domain of self-adjointness $W(\cloud) D(\D\Gamma( h))$. 
 
The inclusion $D_\mathrm{ IBC} \subseteq W(\cloud) D(\D\Gamma( h))$ follows from the symmetry of $(H_\mathrm{ IBC},D_\mathrm{ IBC})$, Proposition \ref{corol:focksymmetry}. 
To show that also $ W(\cloud) D(\D\Gamma( h))\subseteq D_\mathrm{ IBC} $, we use 
that  $W(\cloud)\,D(\D\Gamma( h))$ is the closure of $W(\cloud)E(\Hz(\R^3))$ 
in the graph norm of $W(\cloud)\D\Gamma(h)W(-\cloud) $. 
We need to show that for $\psi\in W(\cloud)\,D(\D\Gamma( h))$ we have $\psi^{(n)}\in D(\lapadjn)$ and 
 $A\psi\in \Fock$. Let $u\in \Hz(\R^3)$, then we have the estimate
\begin{align}
\Vert u(0)W(\cloud)\coh(u)\Vert_\Fock^2
&=\sum_{n\geq 0} \frac{1}{n!} \Vert u(0) u^{\otimes n}\Vert^2_{\Hio^n}\nonumber\\
&\leq \sum_{n\geq 0} \frac{C}{(n+1)!}  (n+1)\Vert (-\Delta_{x_{n+1}} + E_0) u^{\otimes (n+1)}\Vert^2_{\Lz(\R^{3(n+1)})}\nonumber\\
&\leq C\Vert{\D\Gamma(h)\coh(u)}\Vert_{\Fock}^2\,,
\end{align}
where we have used that $|u(0)|\leq C \Vert u\Vert_{\Hz}$ and that $\langle \Delta_{x_j} u^{\otimes (n+1)}, \Delta_{x_i} u^{\otimes (n+1)}\rangle\geq 0$.
In view of Equation~\eqref{AW} this implies that
\begin{equation}\label{eq:Arelbound}
\Vert  A W(\cloud)\coh(u) \Vert_\Fock \leq C\Vert{\D\Gamma(h)\coh(u)}\Vert_{\Fock} 
\end{equation}
for some constant $C>0$. This clearly implies that for any $n\in \N$
\be 
 \Vert (-\lapadjn + n E_0)\left(W(\cloud)\coh(u)\right)^{(n)}\Vert_{\Hio^n}
 \leq \Vert (H-gA) W(\cloud)\coh(u) \Vert_\Fock \leq C \Vert{\D\Gamma(h)\coh(u)}\Vert_{\Fock}\,.
\ee 
As $\lapadjn$ is closed, it follows that $W(\cloud)\,D(\D\Gamma( h))\vert_{\Hio^n}\subset D(\lapadjn)$.

Consequently by Lemma~\ref{lem:ABcont} the expressions for $A$ and $B$ are well defined (as distributions) and continuous on each sector of $W(\cloud)\,D(\D\Gamma( h))$. Now~\eqref{eq:Arelbound} implies that $A$ maps $W(\cloud)\,D(\D\Gamma( h))$ to $\Fock$, so in particular $A\psi^{(n)}\in \Lz(\R^{3n-3})$.
Since $B\psi=g\psi$ on the dense set $W(\cloud)\,E(\Hz)$, this also holds on $W(\cloud)\,D(\D\Gamma( h))$ by continuity, and we have proved $W(\cloud)\,D(\D\Gamma( h))\subset D_\mathrm{ IBC}$.
\end{proof}

We remark that the expressions $A$ and $B$ defined on some natural domain $D\subset \bigoplus_{n} D(\lapadjn)$ are not necessarily closable, e.g., $B$ vanishes on the dense (in $\Fock$) subspace $D(\D\Gamma( h))$, so we cannot directly conclude from an estimate such as~\eqref{eq:Arelbound} that these expressions are well defined on the closure of $W(\cloud)E(\Hz)$.  \\ \ \\
By virtue of the unitary equivalence, we can compute the ground state of $H_\mathrm{ IBC}$ explicitly, provided $E_0 > 0$. The unique ground state  of the free field $\D\Gamma\left(h\right)$ is the vector $\Omega_0 := (1,0,0,\dots) \in \Fock$, which is called  the Fock vacuum. With $\cloud = g f_{\gamma = \sqrt{E_0}}$ we conclude that $\psi_{min}:=W(\cloud)\Omega_0$ is the unique ground state of $H_\mathrm{ IBC}$
with ground state energy $ \frac{g^2 \sqrt{E_0}}{4\pi}$, i.e.
\be
H_\mathrm{ IBC} \, \psi_{min}=  \frac{g^2 \sqrt{E_0}}{4\pi}\,\psi_{min} \,.
\ee 
 Note that because of $\Omega_0 = \coh(0)$ we can calculate $\psi_{min}$ explicitly by using \eqref{eq:dressingdef},
\be 
\psi_{min} = W(\cloud )\Omega_0 = W(\cloud ) \coh(0) = \mathrm{e}^{-\frac{\left\Vert \cloud \right\Vert^2}{2}} \ \coh(\cloud)\,.
\ee

\subsection{Renormalization: Proof of Theorem~\ref{thm:renorm}}
\label{subsec:renorm}
 Let $h = (-\Delta + E_0, \Hz(\R^3))$, where we now assume that $E_0>0$. This operator is self-adjoint and invertible. In Section~\ref{sec:reno} we defined $W_n:=W(g h^{-1} \seq_n)$ where $\seq_n$ is any sequence of elements of $\Lz(\R^3)$ such that $\seq_n\to\delta$ as $n\to\infty$ in the sense that $\hat \seq_n \rightarrow  \hat \seq_\infty:=  \hat{\delta} = (2\pi)^{-3/2}$ pointwise with $\|\hat \seq_n\|_\infty$ uniformly bounded. \\
We first use Corollary \ref{corol:ifinvertible} with $\psi = g \seq_n$ and $T = h$ to establish that, in the notation of Section~\ref{sec:reno},
\begin{align}\nonumber
 W_n H_n W_n^* &=  W_n \left(\D \Gamma(h) + a^*(g \seq_n) + a(g \seq_n) \right) W_n^* =\D \Gamma(h) - g^2 \langle \seq_n , h^{-1} \seq_n \rangle_{\Hio} 
 \\
 &= \D \Gamma(h) + E_n\,.
\end{align}
The assumptions we made on the sequence $\seq_n$ imply that $ \Fou \, (g h^{-1} \seq_n) $ converges in $\Lz$ to the function $g  (2 \pi)^{-3/2} \hat{h}^{-1}$. Therefore, according to \eqref{eq:fourieroff}, $g h^{-1} \seq_n$ converges to $- g f_{\sqrt{E_0}}$. 
We have defined the family of unitary operators $W(\ph)$ in \eqref{eq:dressingdef} via coherent vectors. From this  definition it follows that the mapping $\ph \mapsto W(\ph) \psi$ is continuous because the mapping $\ph \mapsto \co \ph)$ is. 
As a consequence, the $W_n$ converge strongly, and the limiting operator is 
\be W_\infty = \lim_{n \rightarrow \infty} W_n = \lim_{n \rightarrow \infty} W(g h^{-1} \seq_n) = W({\textstyle \lim_{n \rightarrow \infty}} \ g h^{-1} \seq_n) = W(- g f_{\sqrt{E_0}}) \, . \ee
Moreover, for any $z\in\C\setminus \R$ also
\begin{align}\nonumber
\lim_{n \rightarrow \infty} (H_n'-z)^{-1} &= 
\lim_{n \rightarrow \infty} W_n^* (\D \Gamma(h)-z)^{-1} W_n = W_\infty^* (\D \Gamma(h)-z)^{-1} W_\infty\\ &=  (W_\infty^*\D \Gamma(h)W_\infty-z)^{-1}
\end{align}
converges strongly because $\sup_n\|W_n^*\|=1$.
Recalling the definition  \eqref{eq:Hhatdef} of $H_\infty$, we find that
\begin{align}
H_\infty&:= \lim_{n\to\infty} H_n'  = W_\infty^*\D \Gamma(h)W_\infty =  W(g f_{\sqrt{E_0}}) \, \D \Gamma( h) \,W(-g f_{\sqrt{E_0}})   \nonumber \\
& \stackrel{\eqref{euqality1}}{=}   H_\mathrm{ IBC} - \tfrac{g^2 \sqrt{E_0}}{4 \pi} 
\end{align}
on $  W(g f_{\sqrt{E_0}}) D(\D\Gamma(h)) =  D_\mathrm{ IBC}$. 

Since $E_0>0$, it follows from Corollary \ref{corol:form-domain} that $D_\mathrm{ IBC} \cap D(\D\Gamma({h}^{1/2})) = \lbrace 0 \rbrace$. We have proven Theorem \ref{thm:renorm}.

\section{Variants of the Model}
\label{sec:generalized}
Throughout this section, let $E_0 >0$ and $N \in \mathbb{N}$ be fixed. We will use the notation that has been introduced in Section \ref{sec:variants} and in particular assume the condition \eqref{vcond1}. Here we will properly define $\tilde{D}_\mathrm{ IBC}$ and prove Theorems \ref{thm:variant1} and \ref{thm:variant2}.

 Observe that $w_i^\lambda \in \Ddelta$ and that $\lapadj w^\lambda_i = \lambda w^\lambda_i$ for $1 \leq i \leq N$, cf.\ \cite{AGHH88}. It is known that that the maps $\psi \mapsto A_i \psi$ and $\psi \mapsto B_i \psi$ define continuous linear functionals on $\Ddelta$.
Furthermore, using a partition of unity, the degree of non-symmetry of $\lapadj$ may be expressed with their help:
\begin{equation}\label{eq:asymmn1-multi}
\langle \ph, - \lapadj \psi \rangle_{\Hio} - \langle -\lapadj \ph, \psi \rangle_{\Hio} = \sum_{i=1}^N \langle B_i \ph, A_i \psi \rangle_\C - \langle A_i \ph , \B_i \psi \rangle_\C \, . 
\end{equation}
Note the following: The set $U(v) :=\bigcap_{i=1}^N\mathrm{ker}X_i $ is a subspace of $\Ddelta$, which is $\Lz$-dense. By further inspection $X_i(\psi)=0$ for all $ 1 \leq i \leq N$ is identified with the conditions that specify the domain of point interactions centered in $\xi_1, \dots, \xi_N$ with parameters $a_i =\frac{\alpha_i}{ \beta_i}$, where $\beta_i=0$ formally corresponds to $a_i= \infty$, see \cite{DFT08}.  

The matrix $S(\lambda)$ is invertible if and only if   $-\lambda$ is \textit{not} an eigenvalue of the point-interaction operator $(-\lapadj , U(v))$, see Theorem $\text{II}.1.1.4$ in \cite{AGHH88}. The number of eigenvalues of this operator is finite, and all its eigenvalues are negative and situated below the essential spectrum, which covers the non-negative real axis.  
That implies, in particular, that for all $E_0 > 0$ and for all admissible choices of $v$ there exists $\lambda > 0$ such that $S(\lambda)$ is invertible.
\begin{lemma}
\label{lem:phi}
Let $v$ obey the condition \eqref{vcond1} and let $(1,1, \dots, 1)^T \in \mathrm{ ran}\, S(\lambda)$. Then there exists $\cloud=\cloud(\lambda) \in \Ddelta$ with the properties
\begin{align}&  \lapadj\cloud = \lambda \cloud \label{eq:property1} \\
& X_k(\cloud) =1 \qquad 1\leq k \leq N \label{eq:property2} \, . \end{align}
\end{lemma}
\begin{proof}
For every choice of $c_1,\ldots,c_N\in \C$ the sum $\sum_{l=1}^N c_l w^\lambda_l$ is an eigenvector of $\lapadj$ with eigenvalue $\lambda$. To obtain \eqref{eq:property2}, we first compute
\begin{align} 
&X_k\left(\textstyle \sum_{l=1}^N c_l w^\lambda_l \right) = \sum_{l=1}^N c_l X_k(w^\lambda_l) =  \sum_{l=1}^N c_l \alpha_k \E^{\I \theta_k} B_k(w^\lambda_l) + c_l \beta_k \E^{\I \theta_k} A_k(w^\lambda_l)\nonumber \\
&= \sum_{l=1}^N  c_l \alpha_k \E^{\I \theta_k} \delta_{kl} +  c_l \beta_k \E^{\I \theta_k} (\delta_{kl} \frac{\sqrt{\lambda}}{4 \pi} + (1-\delta_{kl}) G^\lambda _{kl}) = \sum_{l=1}^N S_{kl} c_l  \, . 
\end{align}
Since $(1,1, \dots, 1)^T \in \mathrm{ ran}\, S(\lambda)$, there are numbers $c_l \in \C$ such that $\sum_{l=1}^N S_{kl} c_l =1$ for all $1 \leq k \leq N$.
Then we set $\phi:= \sum_{l=1}^N c_l w_l$.
\end{proof}
\begin{lemma}
\label{lem:rewrite1}
Let v obey the condition \eqref{vcond1}. Then the degree of non-symmetry of $\lapadj$ can be expressed using  $X_i$ and $Y_i$: for $\ph,\psi\in D(\lapadj)$,
\begin{equation}
\langle \ph, - \lapadj \psi \rangle_{\Hio} - \langle -\lapadj \ph, \psi \rangle_{\Hio} = \sum_{i=1}^N \langle X_i \ph, Y_i \psi \rangle_\C - \langle Y_i \ph , X_i \psi \rangle_\C \, . 
\end{equation}
\end{lemma}
\begin{lemma}
\label{lem:rewrite2}
Let $ \psi \in U(v) = \bigcap_{i=1}^N\mathrm{ker}X_i $ and let $\cloud(\lambda) \in \Ddelta$ with the properties \eqref{eq:property1} and \eqref{eq:property2}. Then
\begin{align*}
&\mathrm{(a)} \quad \sum_{i=1}^N Y_i(\psi) = \langle \cloud, (-\lapadj + E_0) \psi \rangle_{\Hio} - \langle (-\lapadj + E_0) \cloud,  \psi \rangle_{\Hio} \\
&\mathrm{(b)} \quad \sum_{i=1}^N Y_i(\cloud) \in \R \, .
\end{align*}
\end{lemma}
The proofs can be found in the Appendix  \ref{proof:rewrite}. \\
As mentioned above, the operator $\tilde h = (-\lapadj+E_0,U)$ is self-adjoint and is called the $N$-center point-interaction with energy offset $E_0>0$. The coherent domain $E(U)$ is a core of $\D \Gamma(\tilde h )$, see Proposition \ref{prop:esasecquant}.
Next we turn to another subset of $\Ddelta$, which is an affine subspace. If $(1,1, \dots, 1)^T \in \mathrm{ ran}\, S(\lambda)$, define 
\be M=M(\lambda):= \lbrace \ph \in \Ddelta \vert \ph = \cloud(\lambda) + \psi \, , \psi \in U(v) \rbrace \, . \ee
Since $U(v)$ is $\Lz$-dense, so is $M(\lambda)$ and therefore the coherent domain over $E(M)$ is a dense subspace of the symmetric Fock space $\Fock$. Set $\tilde{D}_\mathrm{ IBC} :=E(M)$.
Then on $\tilde{D}_\mathrm{ IBC}$ we find 
\be Y^\Fock_i(\varepsilon(\ph)) = Y_i(\cloud+\psi) \varepsilon(\ph) = (Y_i(\cloud) + Y_i(\psi)) \varepsilon(\ph)\,  \ee 
and 
\be X^\Fock_i(\varepsilon(\ph)) = X_i(\cloud+\psi) \varepsilon(\ph) = X_i(\cloud) \varepsilon(\ph) = \varepsilon(\ph) \, . \ee 
 We are now in a position to define the operator $(\tilde H_\mathrm{ IBC},\tilde D_\mathrm{ IBC})$ which depends on the set of parameters $(v,E_0)$ where $v$ obeys the relation \eqref{vcond1}:
\be \tilde H_\mathrm{ IBC}:=\D \Gamma(-\lapadj + E_0) + \textstyle  \sum_{i=1}^N Y_i^\Fock \quad \text{on} \quad \tilde D_\mathrm{ IBC} := E(M)  \, . \ee
\begin{proof}[\textbf{ Proof of Theorem~\ref{thm:variant1} and Theorem~\ref{thm:variant2}}]
Let $\psi \in U$. Choose $\lambda >0$ such that $S(\lambda)$ is invertible and use $(1,1, \dots, 1)^T \in \mathrm{ ran}\, S(\lambda)$ to construct $\cloud(\lambda)$ with the properties \eqref{eq:property1} and \eqref{eq:property2}.  
Due to property \eqref{eq:property1} of $\cloud=\cloud(\lambda)$, using Lemma \ref{lemma:pulltrough} we get
\begin{align}
 W(-\cloud)& \D \Gamma(-\lapadj+ E_0) W(\cloud) \co \psi) + W(-\cloud) \left[\textstyle  \sum_{i=1}^N Y_i^{\Fock} \right] W(\cloud) \co \psi)\nonumber \\
 =& \left(\D\Gamma(\tilde h) + (-\lambda + E_0)  \left( a(\cloud)^* +a(\cloud) \right) \right) \co \psi)  \nonumber\\
&   + \left(\langle \cloud, (-\lapadj + E_0) \psi \rangle_{\Hio} - \langle (-\lapadj + E_0) \cloud , \psi \rangle_{\Hio} \right) \co \psi)  \nonumber\\
&   + \langle \cloud, (-\lapadj + E_0) \cloud \rangle_{\Hio} \co \psi) + \left[\textstyle  \sum_{i=1}^N Y_i(\cloud) + Y_i(\psi) \right]  \co \psi)  \nonumber\\
  =&\left(\D\Gamma(\tilde h) + (-\lambda + E_0)  \left( a(\cloud)^* +a(\cloud) \right) \right) \co \psi) \nonumber \\
&   + \left[ (-\lambda+ E_0) \ \Vert \cloud \Vert_{\Hio}  + \textstyle  \sum_{i=1}^N Y_i(\cloud) \right]  \co \psi)  \,.
\end{align}
We have used statement (a) of Lemma \ref{lem:rewrite2}. Due to statement  (b)  of this lemma, the constant in brackets is real. Because $\tilde h$ is bounded from below, we can use Nelson's Commutator Theorem to show essential self-adjointness of the operator on $E(U)$, cf.\ Proposition \ref{prop:moregeneral} and \cite{Der03}. Now essential self-adjointness of $\tilde{H}_\mathrm{ IBC}$ on $W(\cloud(\lambda)) E(U) = E(M) = \tilde{D}_\mathrm{ IBC}$ follows. 

If $(1,1, \dots, 1)^T \in \mathrm{ ran}\, S(E_0)$, set $\lambda=E_0$ to get \eqref{unitaryeq}. We have proven Theorem~\ref{thm:variant2}. In this case $\tilde{H}_\mathrm{ IBC}$ may be unbounded from below. 

If $\tilde h$ is strictly positive, then $-E_0$ is not an eigenvalue of $(-\lapadj, U)$ and $S(E_0)$ is invertible. From the explicit form \eqref{unitaryeq} we see that, because $\D\Gamma(\tilde h)$ is strictly positive as well, $\Omega_0$ is the unique ground state of $\D\Gamma(\tilde h)$. As a consequence $\tilde{H}_\mathrm{ IBC}$ is bounded from below by 
\be C(\cloud(E_0)) = \sum_{i=1}^N Y_i(\cloud(E_0))\ee 
and 
\be \psi_{min} = \E^{-\frac{\Vert \cloud(E_0) \Vert^2}{2}}  \co \cloud(E_0))
\ee is the unique ground state of $\tilde H_\mathrm{ IBC}$.
\end{proof} 
 \begin{appendix}
%\section{Proofs}
 
\section{Regularity}
\label{sec:appreg}
Here, we give the details on the regularity questions regarding $D(\Delta^*_n)$, $A^{(n)}$, and $B^{(n)}$.
We will need to work with   Hilbert-space-valued distributions.
Keep in mind for the following  that for defining distributions the removal of a point $\{0\}$ from $\R^3$ or the sets $\Col^n$ from $\R^{3n}$ matters, while 
 $\Lz(\R^3\setminus\{0\},X) = \Lz(\R^3 ,X)$ and  $\Lz(\R^{3n}\setminus\Col^n,X) = \Lz(\Rminusn,X) = \Lz(\R^{3n} ,X)$.
\begin{lemma}\label{lem:domain}
Let $\ph\in D(\lapadjn)$ and equip this space with the graph norm. Then for $j=1,\ldots, n$
\be
\Delta_{x_j} \ph \in \Lz\left(\R_{x_j}^3 , \Hzm (\Rminusnme) \right)\,,
\ee
where $\ph$ is regarded as a vector valued distribution on $\R_{x_j}^3\setminus\{0\}$ and  $\Delta_{x_j}$ is the Laplacian of distributions on that domain taking values in $\Hzm$. 
Moreover, 
\be
\|\Delta_{x_j} \ph \|_{\Lz(\R^3,\Hzm)} \leq \sqrt{2} \|\ph\|_{D( \Delta_n^*)}\,.
\ee
\end{lemma}
\begin{proof}
We will show the case $j=1$. Recall that $D(\Delta_n)=\Hzo(\Omega_n)$. For any $\ph\in D(\lapadjn)$, the map
\begin{equation}\label{eqn:deltaxjdef}
\Delta_{x_1}\ph: \Hzo(\Omega_n)\to \C\,,\quad \psi\mapsto \langle \ph,\Delta_{x_1} \psi\rangle =\langle \lapadjn \ph,\psi\rangle - \sum_{i=2}^N \langle \ph, \Delta_{x_i} \psi\rangle
\end{equation}
extends by density to a  bounded linear functional on the Bochner space\\ $\Lz\left(\R_{x_1}^3, \Hzo(\Rminusnme) \right)$, i.e.,
\begin{equation}\label{eqn:Deltaxj}
\Delta_{x_1}\ph \; \in\;  \Lz\left(\R_{x_1}^3 , \Hzo(\Rminusnme) \right)' \,.
\end{equation}
Since  $\Hzm(\Rminusnme):= \Hzo(\Rminusnme)'$ and this space is reflexive,
we obtain that $\Delta_{x_1}\ph \in \Lz\left(\R_{x_1}^3 , \Hzm(\Rminusnme) \right)$.  It remains to show that  this $\Delta_{x_1}\ph$   is in fact also the Laplacian of $\ph$ in the sense of $\Hzm$-valued distributions, i.e.~that
for all $\phi\in C^\infty_0(\R^3\setminus \{0\})$ and $\xi\in \Hzo(\Rminusnme)$ we have
\be
 \Delta_{x_1}\ph (\phi \xi) = \int_{\R^3}  ( \ph(x ), \xi )_{(\Hzm,\Hzo)}\Delta \phi(x ) \,\D x\,. 
 \ee
The left hand side is by its definition~\eqref{eqn:deltaxjdef}
\be
 \Delta_{x_1}\ph (\phi \xi) = \langle \ph,(\Delta  \phi) \xi \rangle_{\Lz(\R^{3n })}
\ee
and the right hand side is
\begin{align}
 \int_{\R^3}  ( \ph(x ), \xi )_{(\Hzm,\Hzo)}\Delta \phi(x ) \,\D x& =  \int_{\R^3}  \langle \ph(x) , \xi \rangle_{\Lz(\R^{3n-3})}  \,\Delta \phi(x )  \,\D x \nonumber\\
 &=\langle \ph,(\lapo  \phi) \xi \rangle_{\Lz(\R^{3n })}\, ,
 \end{align}
where we made use of the fact that $\ph \in \Lz(\R^3_{x_1} , \Lz(\R^{3n-3}))$.
From \eqref{eqn:Deltaxj} we conclude
\begin{eqnarray}
\|\Delta_{x_1} \ph \|_{\Lz(\R^3,\Hzm)} & =& \sup_{  \|\psi\|_{\Lz(\R^3, \Hzo) }=1}\left(  \|\lapadjn \ph\|_{\Lz} \|\psi\|_{\Lz} + \|\ph\|_{\Lz} \|\psi\|_{\Lz(\R^3, \Hzo) }\right)\nonumber\\
&\leq &  \|\lapadjn \ph\|_{\Lz}  + \|\ph\|_{\Lz}  \leq \sqrt{2} \|\ph\|_{D( \lapadjn)} \, . \qedhere
\end{eqnarray}
\end{proof}
 
\begin{proof}[\textbf{Proof of Lemma \ref{lem:ABcont}}] For clarity, we use the notation $A^{(n)}$ and $B^{(n)}$ in this proof  for the operators on $D(\lapadjn)\subset \Lz(\R^{3n})$.
The case $n=1$ has been proved in Proposition~\ref{prop:asymn1} and we will use it here to show continuity of $A^{(n)}$ and $B^{(n)}$ for $n\geq 2$. Our proof basically follows ideas for the construction of distribution-valued trace maps on Sobolev spaces, as presented, e.g, in \cite{Lio72}.

Define the space
\be 
\Ddhmz := \lbrace \ph \in \Lz(\R^3 ,\Hzm(\Rminusnme)) \vert \Delta_{x}  \ph \in \Lz(\R^3 ,\Hzm(\Rminusnme))  \rbrace\, ,
\ee 
where $\Delta_{x}$ denotes the Laplacian on vector-valued distributions on $\R^3 \setminus \{0 \}$, and
\be 
\Vert \ph \Vert^2_{\Ddhmz} := \Vert \ph \Vert_{\Lz(\R^3 ,\Hzm)}^2 +\Vert \Delta_{x} \ph \Vert_{\Lz(\R^3 ,\Hzm)}^2 \, .
\ee
Then, by Lemma \ref{lem:domain}, we have the continuous injection
\be 
D(\lapadjn) \hookrightarrow \Ddhmz \, .
\ee
 We will show that $A^{(n)}$ is continuous on $\Ddhmz$, which of course implies continuity on $D(\lapadjn)$.  To do so, we approximate any $\ph \in \Ddhmz$ by a sequence $\ph_N$ in the following way:
Let $(\eta_k)_{k\in\N}$ be a complete orthonormal set in $\Hzm(\Rminusnme)$ and 
set $\ph_k(x)  := \langle \eta_k,\ph(x,\cdot)\rangle_{\Hzm} $.
Clearly, because $\ph \in \Lz(\R^3 ,\Hzm(\Rminusnme))$, it holds that
\be 
\sum_{k=1}^N \ph_k(x) \eta_k := \ph_N(x) \stackrel{N\to \infty}{\rightarrow} \ph(x)
\ee
pointwise in $\Hzm$ and   by dominated convergence  in $\Lz(\R^3 ,\Hzm(\Rminusnme))$. 
Now let $\psi \in C_0^\infty(\R^3 \setminus \{ 0\})$ and observe that, because $ \langle \eta_k, \, \cdot \, \rangle_{\Hzm}$ is continuous on $\Hzm$ and $ \ph(x,\cdot) \Delta \psi$ is integrable, we have that
\begin{align}
&\int_{\R^3}  \ph_k \Delta\psi \, \D x =\int_{\R^3} \langle \eta_k ,  \ph(x,\cdot) \Delta \psi \rangle_{\Hzm} \, \D x = \left\langle \eta_k , \int_{\R^3} \ph(x,\cdot) \Delta \psi \, \D x \right\rangle_{\Hzm} \nonumber\\
&= \left\langle \eta_k , \int_{\R^3}  \psi\Delta_x \ph(x,\cdot) \, \D x \right\rangle_{\Hzm} 
= \int_{\R^3} \langle \eta_k , \Delta_{x} \ph(x, \cdot) \rangle_{\Hzm} \psi(x) \, \D x \,.
\end{align} 
Since $\ph \in \Ddhmz$, $\langle \eta_k , \Delta_{x} \ph(x, \cdot) \rangle_{\Hzm} \in \Lz(\R^3)$ and thus $\ph_k \in \Ddelta$ with $\Delta_1^* \ph_k=\langle \eta_k , \Delta_{x} \ph(x, \cdot) \rangle_{\Hzm}$.

To prove that the limit in the expression for $A^{(n)}$ exists, let 
\be 
\tilde\ph_k(r):= \frac{1}{4\pi}\int_{S^2} r \ph_k(r\omega) \D \omega\,.
\ee 
One easily sees that $\Vert \tilde\ph_k \Vert_{\Hz((0,\infty))}= \Vert \ph_k \Vert_{\Ddelta}$, and thus $\tilde \ph_k$ has a representative in $C^{1,\frac14}([0,\infty))$. More precisely, the Fourier inversion formula yields
\be 
\vert \tilde\ph_k'(R)-\tilde\ph_k'(r) \vert\leq \frac{1}{\sqrt{2\pi}} \left\Vert\frac{k(\E^{\I k R} - \E^{\I k r})}{1+k^2}\right\Vert_{\Lz(\R)}  \Vert\tilde\ph_k \Vert_{\Hz((0,\infty))}\
\leq \delta(R,r) \Vert \ph_k \Vert_{\Ddelta}\,,
\ee 
where $\ph_k'$ denotes the derivative of $\ph_k$.
Then we also have that
\begin{align}
 &\left\Vert \sum_{k=1}^\infty (\tilde \ph_k'(R)-\tilde\ph_k'(r))\eta_k \right\Vert^2_{\Hzm}
 = \sum_{k=1}^\infty \vert \tilde\ph_k'(R)-\tilde\ph_k'(r) \vert^2\nonumber\\
 &\leq \delta(R,r) \sum_{k=1}^\infty \Vert \ph_k \Vert_{\Ddelta}^2 \nonumber\\
 &= \delta(R,r) \sum_{k=1}^\infty  \left(\Vert \ph_k\Vert^2_{\Lz(\R^3)}  + \Vert \Delta_1^* \ph_k\Vert^2_{\Lz(\R^3)}\right)\nonumber\\
 &= \delta(R,r)\left(\sum_{k=1}^\infty   \Vert \ph_k \eta_k \Vert^2_{\Lz(\R^3 ,\Hzm)}  + \sum_{k=1}^\infty  \Vert \Delta_x \ph_k \eta_k \Vert^2_{\Lz(\R^3 ,\Hzm)}\right)\nonumber\\
 &= \delta(R,r) \Vert \ph \Vert^2_{\Ddhmz}\,.
\end{align}
Since for $R,r<1$ we have that $\delta(R,r)\leq C(R-r)^{1/4}$, it follows that the limit $\lim_{r\to 0} \sum_{k=0}^\infty \tilde \ph_k'(r) \eta_k$ exists for this representative of $\ph$ and yields the value of $A^{(n)}/\sqrt{n}$. In addition, we have that 
 \begin{align}
\Vert A^{(n)} \ph_N \Vert_{\Hzm}^2  &= \left\Vert  A^{(n)} \sum_{k=1}^N    \ph_k    \eta_k \right\Vert_{\Hzm}^2  = n \left\Vert    \sum_{k=1}^N    (A^{(1)}\ph_k)    \eta_k \right\Vert_{\Hzm}^2  
\nonumber\\
&
\leq n \Vert A^{(1)}\Vert^2_{\Ddelta'} \sum_{k=1}^N \Vert \ph_k \Vert^2_{\Ddelta} 
\nonumber \\
 & = n \Vert A^{(1)}\Vert^2_{ \Ddelta'} \Vert \ph_N \Vert^2_{\Ddhmz} \,.
 \end{align}
 Thus, $A^{(n)}$ defines a bounded linear map. The proof for $B^{(n)}$ follows the same steps. 
 
 This proof shows that the action of $A^{(n)}$, $B^{(n)}$ is determined by the action of $A^{(1)}$, $B^{(1)}$ on the $\varphi_k$. If $\varphi$ is an element of $\Hz(\R^{3n})$ or $\mathrm{H}^1(\R^{3n})$, then the $\varphi_k$ are in the corresponding space over $\R^3$. In case $\varphi \in \mathrm{H}^1(\R^{3n})$ we thus have that $B^{(n)}\varphi=0$ since $B^{(1)}=0$ on $\Ddelta\cap \mathrm{H}^1(\R^{3})=\Hz(\R^3)$ because $f_\gamma \notin \mathrm{H}^1(\R^3)$. If $\varphi\in \Hz(\R^{(3n)})$, $A^{(n)}$ acts as the Sobolev-trace, because $A^{(1)}\varphi_k=\varphi_k(0)$.
\end{proof}

In order to establish regularity of the functions $\varphi\in D(\Delta_n^*)$ with $B^{(n)}\varphi=0$, we use a theorem of H\"ormander, which is formulated using the following spaces:
\be
H_{(2,s)} := \Lz([0,\infty), \mathrm{H}^{2+s}(\R^{d}) ) \cap \Hz((0,\infty), \mathrm{H}^{s}(\R^d))\,.
\ee

\begin{thm}\label{thm:hoermander}
Let $\tfrac{\D^2}{\D r^2}$ and $\Delta_{\R^d}$ denote the distributional Laplacians on $(0, \infty)$ and $\R^d$, respectively. The map
\begin{eqnarray}
H_{(2,s)} &\to& \Lz([0,\infty), \mathrm{H}^{s})\oplus \mathrm{H}^{s+\frac32}(\R^d)\,,\nonumber\\
\eta &\mapsto & \left( \left(\tfrac{\D^2}{\D r^2} + \Delta_{\R^d}-1\right)\eta,\, \eta(0)\right)
\end{eqnarray}
is an isomorphism of topological vector spaces.
\end{thm}
This theorem is a direct consequence of \cite[Corollary~10.4.1]{Hoer64}. It gives rise to the following regularity lemma, where we denote by  $P:\Lz(\R^3) \to \Lz(\R^3)$ the projection to the space of radial functions; for $j\in \{1, \dots, n\}$, $P_j$ is the projection $P$ acting on the $j$-th factor of $\Lz(\R^{3n})=\Lz(\R^3)^{\otimes n}$; and $Q_j=1-P_j$.
\begin{lemma}\label{lem:basic}
Let $\ph\in D(\lapadjn)$ with $B\ph = 0$ and $\chi_\epsi \in C^\infty_\mathrm{ b}(\R^{3n-3})$ such that, for some $\epsi>0$,
\be 
 \mathrm{ supp}\chi_\epsi\subset  \mathcal{U}_\epsi(\mathcal{C}^{n-1}) := \Bigl\{ (x_2,\ldots,x_n)\in \R^{3n-3}\,\Big| |x_i|>\epsi \,\mbox{ for all } i\Bigr\}\,.
\ee 
Then $\chi_\epsi  P_1 \ph\in \Hz(\R^{3n})$.
\end{lemma}
\begin{proof}
We assume without loss of generality that $\ph$ is radial in the first argument, i.e., $\ph = P_1\ph$.
Let $\tilde\ph(r,y) := r \,\chi_\epsi (y) \ph(r,y)\in \Lz([0,\infty) , \Lz(\R^{3n-3}))$.
First note that 
\begin{eqnarray}
\Delta \tilde \ph &=& \chi_\epsi  (\tfrac{\D^2}{\D r^2} +\Delta_y) r\ph +     (\Delta_y\chi_\epsi ) r  \ph  + 2 r \nabla_y \chi_\epsi \cdot \nabla_y \ph \nonumber\\
&=& \chi_\epsi  \underbrace{r \lapadjn \ph }_{\in \Lz} +  (\Delta_y\chi_\epsi ) \underbrace{ r  \ph}_{\in \Lz}  + 2   \nabla_y \chi_\epsi \cdot \hspace{-18pt} \underbrace{\nabla_y r \ph}_{\in \Lz([0,\infty), \mathrm{H}^{-1} )} \label{deltaph}
\end{eqnarray}
and that $B\ph=0$ implies $\tilde\ph(0) = 0 \in \Hzm$. This of course means that $\tilde \ph(0) \in \mathrm{H}^{s+\frac32}$ for any $s\in \R$.
Thus, Theorem~\ref{thm:hoermander} implies that
\be
\tilde \ph \in H_{(2,-1)} \subset \Lz([0,\infty) , \mathrm{H}^1(\R^{3n-3}))\,. 
\ee
Plugging this information into Equation~\eqref{deltaph}, we conclude that $\Delta\tilde\ph\in \Lz([0,\infty) , \Lz)$. Another use of Theorem~\ref{thm:hoermander} then yields $\tilde \ph \in H_{(2,0)}$ with $\tilde\ph(0) = 0$. Hence
\be
\frac{\tilde \ph}{r} = \chi_\epsi P_1 \ph \in \Lz(\R^3, \Hz(\R^{3n-3}))\cap \Hz(\R^3, \Lz(\R^{3n-3})) = \Hz(\R^{3n})\,.   
\ee
\end{proof} \ \\
For $I \subset \lbrace 1, 2, \dots , n \rbrace$ define the following sets: 
\be \Col^{I} := \Bigl\lbrace x\in \R^{3n} \;\Big\vert\; \prod_{j \in I} \vert x_j \vert=0 \Bigr\rbrace \, . \ee
Then we have $\Col^{I} \subset \Col^n=\Col^{\lbrace 1, 2, \dots , n \rbrace}$. We will also use the abbreviation $\Col^{k} := \Col^{\lbrace n-k+1 , n-k+2, \dots , n \rbrace}$.

\begin{proof}[\textbf{Proof of Proposition \ref{regprop}}]
We will prove that $\ph \in D(\lapadjn) \cap \Hio^n$ together with $B \ph = 0$ implies $\ph \in \Hz(\R^{3n})$. This will prove the statement when combined with Lemma~\ref{lem:ABcont}.

In this proof we write $D^*(X)$ for the adjoint domain of the Laplacian defined on $X\subset \Hz(\R^{3n})$.
For $I\subset \{1,\ldots,n\}$ let 
$
P_I := \prod_{i\in I} P_i 
$ and $
Q_I := \prod_{i\in I} (1-P_i) 
$.
Then for $f \in \Lz(\R^{3n})$ 
\be
f = \prod_{i=1}^n (P_i + Q_i) f = \sum_{I\subset \{1,\ldots,n\}} P_I Q_{I^c} f\,. 
\ee
Now let $\psi\in \Hz(\R^{3n})$. Then
\be
\langle \ph,\Delta\psi\rangle = \sum_{I\subset \{1,\ldots,n\}} \langle P_I Q_{I^c}\ph,\Delta\psi\rangle
= \sum_{I\subset \{1,\ldots,n\}} \langle P_I\ph,\Delta Q_{I^c} \psi\rangle\,.
\ee
Since $Q_j \psi|_{x_j=0}=0$, we have that $Q_{I^c} \psi \in \Hzo(\R^{3n}\setminus \Col^{I^c})$ (cf.~\cite{Sve81}), and so it is sufficient to show that  
\begin{equation}\label{eq:P_I reg}
P_I   \ph  \in D^*(\Hzo(\R^{3n}\setminus \Col^{I^c}))
\end{equation}
in order to conclude $\ph\in  D^*(\Hz(\R^{3n})) = \Hz(\R^{3n})$.
By symmetry it suffices to consider the sets $I= \{ 1,\ldots, k\}$ for $k\leq n$,
which will be done by induction over $k$. 

For $k=1$, $I=\{1\}$, Equation~\eqref{eq:P_I reg} follows from Lemma~\ref{lem:basic} in the following way: Let $\psi\in \Hzo(\R^{3n}\setminus \Col^{n-1})$ and let $\psi_\epsi$ be a sequence in $C_0^\infty( \R^{3n}\setminus \Col^{n-1})$ with supp$\psi_\epsi\subset \mathcal{U}_{2\epsi}$ converging to $\psi$ in $\Hz$. Then Lemma~\ref{lem:basic} implies
\be
\langle P_1 \ph, \Delta_{x_1} \psi_\epsi\rangle = \langle \chi_\epsi(x_2, \dots ,x_n) P_1 \ph, \Delta_{x_1} \psi_\epsi\rangle 
=  \langle \chi_\epsi \Delta_{x_1} P_1 \ph,  \psi_\epsi\rangle  =  \langle  \Delta_{x_1} P_1 \ph,  \psi_\epsi\rangle\,,
\ee
where we have used  a cutoff $\chi_\epsi$ with $\chi_\epsi\equiv 1$ on $\mathcal{U}_{2\epsi}$. 
 Since $\psi_\epsi\in \Lz\left(\R_{x_1}^3, \Hzo(\Rminusnme) \right)$,  
 we find that
\begin{align}
\langle P_1 \ph, \Delta \psi\rangle 
&\stackrel{\hphantom{(30)}}{=} \lim_{\epsi\to 0} \langle P_1 \ph, \Delta \psi_\epsi\rangle  
\stackrel{\eqref{eqn:Deltaxj}}{=} \lim_{\epsi\to 0} \left( \langle \Delta_{x_1} P_1 \ph,  \psi_\epsi\rangle  +    \langle P_1 \ph,\sum_{j=2}^n \Delta_{x_j} \psi_\epsi\rangle  \right)\nonumber
\\
&\stackrel{\eqref{eqn:deltaxjdef}}{=} \lim_{\epsi\to 0}   \langle \lapadjn P_1 \ph,  \psi_\epsi\rangle 
=  \langle \lapadjn P_1 \ph,  \psi \rangle \,.
\end{align}
Hence, $P_1\ph \in D^*( \Hzo(\R^{3n}\setminus \Col^{n-1}))$.

Now assume the induction hypothesis 
\begin{align}
\label{eq:coldomain1}
P_{\{1,\ldots, k\}} \ph  \in D^*(\Hzo(\R^{3n}\setminus \Col^{\{ k+1,\ldots,n\}})) \, .
\end{align}
By symmetry, the argument for $k=1$ independently gives also 
\begin{align}
\label{eq:coldomain2}
P_{\{ k+1\}} \ph  \in 
D^*( \Hzo( \R^{3n}\setminus \Col^{\{1,\ldots,k,k+2,\ldots,n\}}  ))\,.
\end{align}
Thus, $P_{\{1,\ldots, k+1\}} \ph$ is in the intersection of these two domains \eqref{eq:coldomain1} and \eqref{eq:coldomain2}. Clearly, for two dense domains $D_1,D_2$ it holds that 
$
D^*(D_1)\cap D^*(D_2) \subset D^*(D_1+D_2)
$.
We thus need to show that 
\begin{equation}
\label{sumH20}
\Hzo(\R^{3n}\setminus \Col^{\{ k+1,\ldots,n\}}) +  \Hzo( \R^{3n}\setminus \Col^{\{1,\ldots,k,k+2,\ldots,n\}}  ) 
\end{equation}
is dense in $\Hzo(\R^{3n}\setminus \Col^{\{ k+2,\ldots,n\}})$, as this implies that the adjoint domains are equal. The functions in this sum vanish on 
\be
\widetilde{ \Col} := \left(\Col^{\{k+1\}}\cap \Col^{\{1,\ldots, k\}}\right)\cup \Col^{\{k+2,\ldots, n\}}\,.
\ee
Conversely, any function $f\in C^\infty_0 (\R^{3n}\setminus \widetilde{ \Col})$ can be written as a sum $f=f_1+f_2$ with $f_1\in C_0^\infty(\R^{3n}\setminus \Col^{\{k+1\}})$ and $f_2\in C_0^\infty(\R^{3n}\setminus \Col^{\{1,\ldots,k\}})$. Thus the sum \eqref{sumH20} is dense in $\Hzo(\R^{3n}\setminus \widetilde{\Col})$, but the latter space is equal to $\Hzo(\R^{3n}\setminus \Col^{\{ k+2,\ldots,n\}})$, as $\Col^{\{k+1\}}\cap \Col^{\{1,\ldots, k\}}$ has codimension six, see~\cite{Sve81}.
\end{proof}

\section{Algebraic identities}
\label{proof:rewrite}
\begin{proof}[\textbf{Proof of Lemma \ref{polar}}]
\label{proof:polar}
We will prove the following formula:
\begin{align}
\label{eq:ansatzprod}
\mathrm{Sym}(u_1 \otimes  \dots \otimes u_n)  = \frac{1}{2^n n!} \sum_{\mathbf{j} \in J} \alpha_\mathbf{j} \ v_\mathbf{j}^{\otimes n}\,,
\end{align}
where $J=\lbrace 0, 1 \rbrace^n$ and
\be
v_{\mathbf{j}} = \sum_{k=1}^{n} (-1)^{j_k} u_k\, , \qquad \alpha_{\mathbf{j}} = (-1)^{j_1+\ldots+j_n} \,.
\ee
Note that we may rewrite $v_{\mathbf{j}}^{\otimes n}$ as a sum:
\begin{align}
 v_{\mathbf{j}}^{\otimes n}&=\left(\sum_{k=1}^{n} (-1)^{j_k} u_k \right)^{\otimes n} = \sum_{ \mathbf{k} \in \lbrace 1,\dots,n \rbrace^n}  (-1)^{j_{k_1}+\ldots+j_{k_n}} \ u_{k_1} \otimes \cdots \otimes u_{k_n}\nonumber \\
 &= \begin{aligned}[t]
     &\sum_{\mathbf{k}  \in P}  (-1)^{j_{k_1}+\ldots+j_{k_n}} \ u_{k_1} \otimes \cdots \otimes u_{k_n} \\
     &+ \sum_{\mathbf{k}  \in \lbrace 1,\dots,n \rbrace^n \setminus P} (-1)^{j_{k_1}+\ldots+j_{k_n}} \ u_{k_1} \otimes \cdots \otimes u_{k_n}
    \end{aligned}
  \nonumber \\
 & =: (v_\mathbf{j})_P + (v_\mathbf{j})_{P^C} \,.
\end{align}
Here we have introduced a set $P$ of multi-indices:
\begin{align}
P := \bigl\lbrace x \in \mathbb{N}^n \big\vert \, \exists \sigma \in S_n: \, x=\sigma (1 , 2 ,\dots , n) \bigr\rbrace \subset \lbrace 1 ,2, \dots ,n \rbrace^n   \,.
\end{align}
We will focus on $(v_{\mathbf{j}})_P$ first and insert it into our ansatz \eqref{eq:ansatzprod}:
\begin{align}
&\sum_{\mathbf{j} \in J} \alpha_\mathbf{j} \ (v_\mathbf{j})_P = \sum_{\mathbf{j} \in J} \sum_{\mathbf{k}  \in P} (-1)^{j_1+\ldots+j_n}  (-1)^{j_{k_1}+\ldots+j_{k_n}} \ u_{k_1} \otimes \cdots \otimes u_{k_n} \nonumber\\
&= \sum_{\mathbf{j} \in J} \sum_{\sigma \in S_n} (-1)^{j_1+\ldots+j_n}  (-1)^{j_{\sigma(1)}+\ldots+j_{\sigma(n)}} \ u_{\sigma(1)} \otimes \cdots \otimes u_{\sigma(n)} \nonumber \\
&= \sum_{\mathbf{j} \in J} \sum_{\sigma \in S_n} (-1)^{j_1+\ldots+j_n}  (-1)^{j_{1}+\ldots+j_{n}} \ u_{\sigma(1)} \otimes \cdots \otimes u_{\sigma(n)}   \nonumber\\
&= \vert \lbrace 0 , 1 \rbrace^n \vert \sum_{\sigma \in S_n}  \ u_{\sigma(1)} \otimes \cdots \otimes u_{\sigma(n)}  = 2^n n! \ \mathrm{Sym}(u_1 \otimes \cdots \otimes u_n) \, .
\end{align}
It remains to show that  $\sum_{\mathbf{j}} \alpha_{\mathbf{j}}\, (v_\mathbf{j})_{P^C}=0$:
\begin{align}
\sum_{\mathbf{j} \in J} &\alpha_\mathbf{j} \ (v_\mathbf{j})_{P^C} \nonumber \\
&=\sum_{\mathbf{j} \in J} \ \sum_{\mathbf{k}  \in \lbrace 1,\dots,n \rbrace^n \setminus P} \hspace{-18pt} (-1)^{j_1+\ldots+j_n}   (-1)^{j_{k_1}+\ldots+j_{k_n}} \ u_{k_1} \otimes \cdots \otimes u_{k_n}  \nonumber\\
&= \sum_{\mathbf{k}  \in \lbrace 1,\dots,n \rbrace^n \setminus P}  \left(\sum_{\mathbf{j} \in J} (-1)^{j_1+\ldots+j_n}  (-1)^{j_{k_1}+\ldots+j_{k_n}} \right) u_{k_1} \otimes \cdots \otimes u_{k_n} \,.
\end{align}
 We will show that the  expression in brackets vanishes. 
 For every  $\mathbf{k}  \in \lbrace 1,\dots,n \rbrace^n \setminus P$ there is at least one $m \in \lbrace 1,\dots,n \rbrace$ such that none of the $k_i$ is equal to $m$. Therefore, we can factor out
\begin{align}
& \sum_{\mathbf{j} \in J} (-1)^{j_1+\ldots+j_n}  (-1)^{j_{k_1}+\ldots+j_{k_n}}\nonumber \\
&= \sum_{j_m =0}^1 (-1)^{j_{m}} \sum_{\mathbf{j} \in \lbrace 0,1 \rbrace^{n-1}} (-1)^{j_1+\ldots+ \widehat{j_m}+ \ldots+j_n}  (-1)^{j_{k_1}+\ldots+j_{k_n}}\,,
\end{align}
because the remaining term on the right does not depend on $j_m$ any more. Now $\sum_{j_m} (-1)^{j_m}=0$.
\end{proof}
 \begin{proof}[\textbf{Proof of Lemma \ref{lem:rewrite1}}]
 For $\ph,\psi\in D(\lapadj)$,
\begin{align}
 \langle X_i(\ph), Y_i(\psi) \rangle_\C & - \langle Y_i(\ph), X_i(\psi) \rangle_\C \nonumber\\
 \stackrel{\hphantom{(35)}}{=} &\Bigl\langle \E^{\I \theta_i}(\alpha_i \B_i + \beta_i \A_i)(\ph), \E^{\I \theta_i}(\gamma_i \B_i + \delta_i \A_i)(\psi) \Bigr\rangle_\C \nonumber \nonumber\\
&  - \Bigl\langle \E^{\I \theta_i}(\gamma_i \B_i + \delta_i \A_i)(\ph), \E^{\I \theta_i}(\alpha_i \B_i + \beta_i \A_i)(\psi) \Bigr\rangle_\C  \nonumber\\
 \stackrel{\eqref{vcond1}}{=} &(\alpha_i \delta_i - \beta_i \gamma_i) \langle  \B_i \ph , \A_i \psi  \rangle_\C - (\alpha_i \delta_i - \beta_i \gamma_i) \langle  \A_i \ph , \B_i \psi  \rangle_\C \nonumber \\
 \stackrel{\hphantom{(35)}}{=} & \langle \B_i \ph, \A_i \psi \rangle_\C - \langle \A_i \ph , \B_i \psi \rangle_\C 
\end{align}
because the terms involving twice $\B_i$ or twice $\A_i$ cancel, and only the mixed terms survive. Summing the terms from all sources $i=1,\ldots, N$   yields the claim.
\end{proof}
\begin{proof}[\textbf{Proof of Lemma \ref{lem:rewrite2}}]
By assumption, $X_i(\psi)=0$ and $X_i(\cloud)=1$ for $i=1,\ldots,N$. Thus, from Lemma~\ref{lem:rewrite1} with $\ph=\cloud$, 
\begin{align}
\sum_{i=1}^N Y_i(\psi) &= \sum_{i=1}^N \langle X_i(\cloud), Y_i(\psi) \rangle_\C = \sum_{i=1}^N \langle X_i(\cloud), Y_i(\psi) \rangle_\C - \langle Y_i(\cloud), X_i(\psi) \rangle_\C\nonumber \\
& = \langle \cloud, - \lapadj \psi \rangle_{\Hio} - \langle - \lapadj \cloud, \psi \rangle_{\Hio}\nonumber \\
&= \langle \cloud, (- \lapadj +E_0) \psi \rangle_{\Hio} - \langle (- \lapadj +E_0) \cloud,  \psi \rangle_{\Hio}\,.
\end{align}
This proves statement (a). To see why (b) is also true, observe that,  since by assumption  $\lapadj\cloud=\lambda\cloud$,
\begin{align}
2  \I \, \mathrm{Im}\left(\sum_{i=1}^N Y_i(\cloud)\right) 
&= \sum_{i=1}^N Y_i(\cloud) - \overline{Y_i(\cloud)} = \sum_{i=1}^N \langle X_i(\cloud), Y_i(\cloud) \rangle_\C - \langle Y_i(\cloud), X_i(\cloud) \rangle_\C\nonumber \\
&= \langle \cloud, - \lapadj \cloud \rangle_{\Hio} - \langle -\lapadj \cloud, \cloud \rangle_{\Hio}\nonumber \\[3mm]
&= \langle \cloud, (- \lambda+E_0) \cloud \rangle_{\Hio} - \langle (-\lambda+E_0) \cloud, \cloud \rangle_{\Hio} = 0\,,
\end{align}
which completes the proof.
\end{proof}
\end{appendix}

\paragraph{Acknowledgments.} We thank Marcel Griesemer, Stefan Keppeler, and Andreas W\"unsch for helpful discussions. This work was supported by the German Science Foundation (DFG) within the Research Training Group 1838 \textit{Spectral Theory and Dynamics of Qauntum Systems}.

\end{document}